\documentclass[10pt]{amsart}
\usepackage{a4wide}
\usepackage{amsfonts}
\usepackage{psfrag}
\usepackage{color}
\usepackage[all]{xypic}
\usepackage{hyperref}
\usepackage{amsmath}
\usepackage{amssymb}
\usepackage{amsthm}
\usepackage{graphicx}
\usepackage{mathtools}
\usepackage{tikz}
\usetikzlibrary{shapes,arrows}

\newtheorem{theorem}{Theorem}[section]
\newtheorem{corollary}[theorem]{Corollary}
\newtheorem{lemma}[theorem]{Lemma}
\newtheorem{proposition}[theorem]{Proposition}

\theoremstyle{definition}
\newtheorem{definition}[theorem]{Definition}
\newtheorem{remark}[theorem]{Remark}

\newtheorem*{condition}{Condition}

\newtheorem {example}{Example}
\numberwithin{equation}{section}
\numberwithin{equation}{section}
\numberwithin{equation}{section}

\newcommand{\wA}{\widetilde{A}}
\newcommand{\wX}{\widetilde{X}}
\newcommand{\wY}{\widetilde{Y}}

\newcommand{\Ss}{\mathcal{S}}

\newcommand{\Hh}{\mathcal{H}}
\newcommand{\PP}{\mathbb{P}}
\newcommand{\HH}{\mathbb{H}}

\newcommand{\N}{\mathbb{N}}

\newcommand{\R}{\mathbb{R}}
\newcommand{\Ff}{\mathcal{F}}
\newcommand{\B}{\mathcal{B}}
\newcommand{\cBp}{\mathcal{B}^{\psi}}
\newcommand{\Ll}{\mathcal{L}}
\newcommand{\rrho}{\overline{\rho}}

\newcommand{\cB}{\ensuremath{\mathcal B}}
\newcommand{\cE}{\mathcal{E}}
\newcommand{\cA}{\mathcal{A}}

\newcommand{\dE}{d^{(\cE)}}
\newcommand{\bF}{\bf F}
\newcommand{\bA}{\bf A}
\newcommand{\bPi}{\bf \Pi}
\newcommand{\DcE}{D(\cE)}
\newcommand{\DocE}{D_0(\cE)}
\newcommand{\DcEl}{D(\cE)_{loc}}
\newcommand{\essup}{\textrm{ess}\ \textrm{sup}}
\newcommand{\esinf}{\textrm{ess}\ \textrm{inf}}
\newcommand{\I}{\mathcal{I}}
\newcommand{\hIp}{\widehat{\mathcal{I}}_{\psi}}

\newcommand{\hx}{x}
\newcommand{\hy}{y}
\newcommand{\hX}{X}
\newcommand{\hY}{Y}
\newcommand{\hz}{z}
\newcommand{\hZ}{Z}

\newcommand{\supp}{\textrm{supp}}

\tikzstyle{line} = [draw, -latex']
\tikzstyle{decision} = [ellipse, draw,  
    text width=5em, text centered, rounded corners, node distance=1cm, inner sep=0pt]
\tikzstyle{block} = [rectangle, draw,
    text width=5em, text centered, rounded corners, minimum height=2em]
\tikzstyle{wideblock} = [draw, block, node distance=1cm,
    text width=10em, text centered, minimum height=2em]
\tikzstyle{diamond} = [draw, diamond,fill=red!20, node distance=1cm,
    text width=5em, text badly centered, minimum height=2em]
\tikzstyle{cloud} = [draw, ellipse, node distance=1cm,
    text width=6em, text badly centered, minimum height=2em]
\tikzstyle{together} = [rectangle,
    text width=5em, text centered, node distance=1cm, inner sep=0pt]

\begin{document}
\title[
Functional Analytic (Ir-)Regularity Properties of SABR-type Processes]{
Functional Analytic (Ir-)Regularity Properties of SABR-type Processes
}

\author{Leif D\"{o}ring}
\address{Department of Mathematics,  University of Mannheim}
\email{doering@uni-mannheim.de}

\author{Blanka Horvath}
\address{Department of Mathematics, Imperial College London}
\email{b.horvath@imperial.ac.uk}

\author{Josef Teichmann}
\address{Department of Mathematics, ETH Z\"urich}
\email{josef.teichmann@math.ethz.ch}
\date{March 6, 2016}

\thanks{BH acknowledges financial support from the SNF Early Postdoc Mobility Grant 165248.
}
 
  
\maketitle

\bigskip
 
\begin{abstract}
The SABR model is a benchmark stochastic volatility model in interest rate markets, 
which has received much attention in the past decade. 
Its popularity arose from a tractable asymptotic expansion for implied volatility, derived by heat kernel methods.
As markets moved to historically low rates, this expansion appeared to yield inconsistent prices. Since the model is deeply embedded in market practice, alternative pricing methods for SABR have been addressed in numerous approaches in recent years.
All standard option pricing methods make certain regularity assumptions on the underlying model, but for SABR these are rarely satisfied. 
We examine here regularity properties of the model from this perspective
 with view to a number of (asymptotic and numerical) option pricing methods. In particular, we highlight 
delicate degeneracies of the SABR model (and related processes) at the origin, which deem the currently used popular heat kernel methods and all related methods from (sub-)
~Riemannian geometry ill-suited for SABR-type processes, when interest rates are near zero.
We describe a more general semigroup framework, which permits to derive a suitable geometry for SABR-type processes (in certain parameter regimes)  via symmetric Dirichlet forms. 
Furthermore, we derive regularity properties (Feller- properties and strong continuity properties) necessary for the applicability of popular numerical schemes to SABR-semigroups, and identify suitable Banach- and Hilbert spaces for these. 
Finally, we comment on the short time and large time asymptotic behaviour of SABR-type processes beyond the heat-kernel framework.
\end{abstract}

\bigskip
\tableofcontents

\section{Introduction}
The stochastic alpha, beta, rho (SABR) model is defined by the 
following system of stochastic differential equations
\begin{align}\label{eq:SABRSDE}
dX_t=Y_t X_t^{\beta}dW_t,\qquad
dY_t= \alpha Y_t dZ_t, \qquad
d\langle Z,W\rangle_t=\rho dt,
\end{align}
where $X_0=x, Y_0=y, t\geq 0$ and where $W$ and $Z$ are  $\rho$-correlated Brownian motions on a filtered probability space
$(\Omega, \Ff, (\Ff_t)_{t\geq 0}, \PP)$
with parameters $\alpha \geq 0$, $\beta \in[0,1]$ and $\rho \in [-1,1]$ and state space $D:= \R_{\geq 0}\times \R_{>0}$.
It describes the dynamics of a forward rate $X$ with CEV-type (see \cite[Section 6.4]{jeanblanc2009mathematical}) dynamics and a stochastic volatility $Y$.  
The SABR model was introduced by Hagan et. al in \cite{ManagingSmileRisk,HLW} in the early 2000's and is today benchmark in interest rate markets \cite{AntonovFree,Antonov,BallandTran}. 
During its decade and a half of existence, the SABR model has changed the daily routine of interest rate modelling due to a particularly tractable formula (the so-called SABR formula) for the implied volatility. This has led to the formula rather than the model itself becoming an industry standard. Today it is known \cite{Obloj} that the celebrated SABR formula is prone to yield inconsistent prices around zero interest rates, but the model itself can produce consistent prices when appropriate pricing methods are used. In the historically low interest-rate environments of the past years this has become increasingly relevant and spurred research both among practitioners and academics. From an academic perspective, the concrete but delicate examples of SABR-type models exhibit several advantageous but also certain challenging properties. We review here some regularity properties of this family of processes with a view to different option pricing approaches.\\

As a starting point we study the applications of a time change that partially decouples the equations and allows to relate SABR to a CEV model running on a stochastic clock.
The first time-change applied to SABR appeared in the original work \cite{ManagingSmileRisk} as a rescaling. There the volatility of volatility, as a scaling parameter of time, was key---via singular perturbation---to derive the tractable asymptotic formula which made the model popular. 
Later on in \cite{Antonov, Islah}, a further time change was proposed for the SABR model, where not only the volatility of volatility, but the entire volatility path is absorbed into the (now random) rescaling of time.
Time change arguments of this type are standard and appear in different contexts in the finance literature, for example in \cite{Veraart} and \cite{KallsenShiryaev} in a general setup and in \cite{Antonov, ChenOsterlee, FordeLargeMaturity,Hobson} in the context of the SABR model. The geometric viewpoint on the SABR model was put forward in \cite{HLW,Labordere} and in the monograph \cite{LabordereBook}, drawing the link between a \emph{normal} ($\beta=0$) SABR model and Brownian motion on a hyperbolic plane. It provides a technique to re-derive the original singular perturbation expansion of Hagan et al.
via heat kernel techiques. These are applied to Brownian motion on a suitably chosen manifold (the SABR plane, cf. \cite{HLW}) combined with a simpler---regular---perturbation. 
This geometric insight further promoted the integration of tools from stochastic analysis on manifolds \cite{Elworthy,Hsu} into the context of mathematical finance. It also initiated refinements of the initial Hagan expansion, which appeared to lack accuracy in the vicinity of the origin, by refining the leading order \cite{BBF,Obloj} and providing a second order term \cite{Paulot}.
Although the viewpoint on the state space of the SABR process as a Riemannian manifold appears to be well suited to characterize
the absolutely continuous part of the distribution when $X\in (0,\infty)$, it potentially breaks down at the origin (for values $\beta \in (0,1)$ of the CEV parameter).
As a result, Riemannian heat kernel expansion techniques applied to SABR can yield erroneous values in the vicinity of $X=0$ already in leading order.
Although several celebrated results for asymptotic expansions beyond the Riemannian (elliptic) setup are available---see for example  \cite{BenArous1,BenArous2,DFJV1,DFJV2,DeMarcoFriz} for hypoelliptic diffusions---these however are not applicable to SABR, due to the its delicate degeneracy.
The literature on short-time asymptotics for degenerate diffusions beyond the hypoelliptic setup is scarce. In certain cases, the  machinery of symmetric Dirichlet forms \cite{SturmLocal,SturmGeometry, terElstDegenerateElliptic} is suitable for such short-time asymptotic expansions. Motivated by this, we establish here Dirichlet forms for SABR and related processes for a suitable subset of parameter configurations.\\

The paper is organized as follows: We review the time change in \cite{Antonov,Islah}
and suggest a perspective on it as a transformation of the underlying geometry of the state space of the processes.
As applications we prove regularity properties and point out certain irregularities of the SABR model and related processes (Sections 2 and 3). Proofs can be found in Appendix \ref{Sec:Proofs}.\\ 

As a first application we propose (in Section 2) a simple method to prove the Feller-Dynkin property of the SABR process. Feller(-Dynkin) processes and Feller semigroups have a rich theory with an interest of its own, see for instance \cite{BoettcherSchillingWang,vanCasteren, vanCasterenDemuth}. 
However, our main interests in Feller-Dynkin properties are motivated from a numerical point of view:
In \cite{Boettcher2} Monte Carlo methods are proposed under the Feller-Dynkin assumption. 
Furthermore, general convergence theory as developed in Hansen, Ostermann \cite{HansenOsternann} provides a framework, where splitting schemes with optimal convergence orders can be constructed. The applicability of their results to the SABR model requires the considered semigroup to be a strongly continuous contraction semigroup on a suitable Banach space.  Feller properties are needed to show strong continuity in these Banach spaces.
\\
Although the Feller-Dynkin property is well studied in different contexts, available standard results are not applicable to SABR: For example the assumptions of \cite[Theorems 21.11 and 23.16]{kallenberg2002foundations}, \cite[Theorem 4.1]{BoettcherSchillingWang} and \cite[Theorem 1]{MijatovicPistorius} are violated\footnote{  
The Brownian motions in \eqref{eq:SABRSDE} are not independent, the function in the multiplicative perturbation resulting from the time change is unbounded and so are the coefficients of the SDE \eqref{eq:SABRSDE}.
}.
Further methods to derive Feller properties under rather general conditions were proposed in \cite{Boettcher1}. These conditions however are also violated by SABR.
Therefore, we propose an approach to derive Feller-Dynkin properties, which applies to SABR-type processes and to a wider class of stochastic models.
Furthermore, since from a financial perspective it is desirable to consider payoff functions with unbounded growth, we also derive so-called \emph{generalized Feller properties} (cf. \cite{SemigroupPOVSplittingSchemes, RoecknerSobol}) for the considered processes.\\
Having determined Banach spaces on which the semigroups of SABR-type processes are strongly continuous (Section 2.1), we turn to the regularity of these semigroups on Hilbert spaces (Section 2.2).
We construct weighted $\Ll^2$ -spaces, on which we
can associate a strongly continuous symmetric semigroup to these processes. Furthermore, we determine all classes
of parameter configurations, for which the symmetry property of the resulting Dirichlet forms is not violated.
In particular, we highlight that apart from the special cases of parameter combinations where either $\beta=0$, $\beta=1$, $\rho=0$ or $\alpha=0$,  SABR-Dirichlet forms are generically not symmetric on any weighted $\Ll^2$ -space.\\

As a further application of the time change we characterize (in Section 3) the asymptotic behavior of the SABR process and of some related processes for  short- and large-time horizons. To study the large time behavior of the SABR process (Section \ref{Sec:AsymptoticsLargetime}) we write the SABR process as a time change of a simpler process---the CEV process---for which the asymptotic behaviour is well known. If the time change does not level off, then the CEV process reaches its $t\rightarrow \infty$ limit and hence the SABR process hits zero. Otherwise, the SABR process has a non-trivial limit behavior since it is the position of the CEV process at a finite random time.
In \cite{Hobson} similar results are derived for special cases. 
We pursue this line of argumentation for general values of the parameter $\beta$ in the uncorrelated SABR model and for Brownian motion on the SABR plane. The perspective of the time change as a transformation of the underlying geometry allows us to relate the large time behavior of the latter to correlated Euclidean Brownian motions.\\ 
 We emphasize the scope of applications of Dirichlet forms by commenting (Section \ref{Sec:AsymptoticsShorttime}) on the short time behavior of a diffusion closely related to the CEV process, for which Varadhan-type asymptotics fail in the parameter regime $\beta\in [1/2,1)$, see \cite{TerElstDiffusion}. 
 The geometry induced (cf. \cite{SturmGeometricAspect,SturmDiDetMetr}) by the respective Dirichlet forms corresponding to the processes  seems to better reflect the behavior of this process near the singularities than classical Riemannian geometry. 
While in Riemannian geometry the distance between two points is determined---via the Eikonal equation---by the  gradient along a (length-minimizing) geodesic curve connecting them, in the Dirichlet geometry this gradient is replaced by an appropriate energy measure. Finally, we conclude by calculating the energy measures for the SABR and CEV Dirichlet forms and observing that the time-change transforms the ``Dirichlet geometry'' analogously to the Riemannian case.

\subsubsection*{The time change}
A crucial observation for the time change arguments we study here is that the generator of the SABR model \eqref{eq:SABRSDE} factorizes as
\begin{align}\label{eq:GeneratorTimeChangedSABR}
Af(x,y)
&=\frac{1}{2}y^2 (x^{2\beta} \partial_{x,x}^2+2 \rho x^{\beta} \partial_{x,y}^2+\partial_{y,y}^2)f(x,y)=y^2\wA f(x,y),\qquad f\in C_c^{\infty}(D)
\end{align}
for an operator $\wA$, where $y^2$ acts as a multiplicative perturbation. 
The martingale problem for the operator $\wA$ (we recall it in the Appendix, \eqref{eq:SABRMartingaleProblem}) is solved by the law of a process 
with dynamics 
\begin{align}\label{eq:SABRdynamicsTimechanged}
d\widetilde X_t=\widetilde{X}_t^{\beta}dW_t, \qquad d\widetilde Y_t= dZ_t, \qquad
d\langle Z,W\rangle_t=\rho dt,
\end{align}
where $\widetilde X_0= \widetilde x$, $\widetilde Y_0=\widetilde y$, $t\geq 0$.
The multiplicative perturbation \eqref{eq:GeneratorTimeChangedSABR} of generators emphasizes (cf.\cite[Theorem 4.1]{BoettcherSchillingWang} and \cite[Theorem VI.3.7]{Bass}) the following well-known relationship between the processes
\eqref{eq:SABRSDE} and \eqref{eq:SABRdynamicsTimechanged}:
\begin{theorem}[Random time change for SABR]\label{Th:TimechangedSABRVolkonskii}
Let the law of the process $(\wX_t,\wY_t)_{t\geq 0}$ be a solution of \eqref{eq:SABRdynamicsTimechanged}.
Then the following processes coincide in law with the SABR process $(X_t, Y_t)_{t\geq 0}$: 
\begin{align}\label{eq:TimechangedSABR}
        \wX_{ \int_0^t (Y_s)^2 ds }= \wX_{\tau^{-1}_t},\qquad
        \wY_{\int_0^t (Y_s)^2 ds }=\wY_{\tau^{-1}_t},
\qquad t \geq 0,
 \end{align}
where the time change is defined as
\begin{align}\label{eq:Thetimechange}
\tau_t:=\int_0^t  \widetilde Y_s^{-2}\,ds \
\end{align}
up to the hitting time $t<T_0^{\widetilde Y}=\inf\{t:\widetilde Y_t=0\}$.
In particular, the law of \eqref{eq:TimechangedSABR} is a solution of the martingale problem \eqref{eq:SABRMartingaleProblem} for $A$ in \eqref{eq:GeneratorTimeChangedSABR}. 
\end{theorem}
The time change \eqref{eq:Thetimechange} for SABR appears in several related works, for example in \cite{Antonov, FordeSharpTail, SABRMassZero, Islah}.
There is a clear advantage in studying SABR from the point of view of Theorems \ref{Th:TimechangedSABRVolkonskii}. 
In \eqref{eq:SABRdynamicsTimechanged} the full coupling in the coefficients is removed and put into the time change.  An important feature of the SABR time change \eqref{eq:Thetimechange} is that it only involves the second coordinate processes $Y$ and $\widetilde Y$ and therefore some important calculations can be handled. Indeed, it is nothing but the time change of a Brownian motion into a geometric Brownian motion. Furthermore, the time change is a continuous additive functional of Brownian motion and as such it has a representation as a (unique) mixture of Brownian local times \cite[Theorem 22.25]{kallenberg2002foundations}. This suggests a perspective on the time change as a transformation of the underlying geometry, cf. Appendix \ref{Sec:TimeChangeLocalTime}. See also \cite{Alili} on a related matter.

\begin{remark}[Geometric perspective on the time change]\label{rem:GeomPerspLocalTime}
By the time change \eqref{eq:Thetimechange} one changes the time that a particle spends in a small neighborhood $U_x(\epsilon)$ of a point $x$. An alternative perspective is to define a new geometry (intrinsic metric or ``energy" metric) in a way to match the speed of the particle: points where the particle moves slowly are ``far away`` (intervals are stretched) and points, where the particle moves quickly are ``close" (intervals are compressed) in the energy metric. 
This is well illustrated in the special case ($\beta=0$, $\rho=0$) of the SABR model by observing in Theorem \ref{Th:TimechangedSABRVolkonskii} the time changed Brownian process $\widetilde{X},\widetilde{Y}$ in the geometry of the hyperbolic plane: the time-changed Brownian particle moves slowly where $\widetilde Y$ is small, hence the speed measure is large. Indeed, distances on the hyperbolic plane around the horizontal axis explode and boundary points are infinitely far from interior points in the hyperbolic metric. See Appendix  \ref{sec:TimechangeGeometry} for more details on this example and \ref{Sec:TimeChangeLocalTime} for a reminder on time change via local times \cite[II.16]{Borodin}) .
\end{remark}
An analogous statement to Theorem \ref{Th:TimechangedSABRVolkonskii} holds for the so-called
\emph{Brownian motion on the SABR plane} (henceforth SABR-Brownian motion) cf. \cite[Section 3.2]{HLW}. 
This process was first considered in \cite{HLW} and is characterized as a stochastic process, whose law solves the martingale problem (see \eqref{eq:SABRSDEwithDrift} below) corresponding to the Laplace-Beltrami operator $\Delta_g$ of the so-called SABR manifold\footnote{Note that the here the manifold is only a subset of the state space of the process, as the axis $\{(x,y):x=0\}$ is excluded. On this set the instantaneous covariance matrix degenerates and the Riemannian metric is not defined.} $(\mathcal{S},g)$, where the manifold and metric tensor are
\begin{equation}\label{eq:ManifoldMetric}
\mathcal{S}:=(0,\infty)^2 \quad\textrm{and}\quad g(x,y) := \frac{1}{1-\rho^2}\left(\frac{d x^2}{y^2x^{2\beta} } - \frac{2 \rho \ d xd y}{y^2 x^{\beta}}
 + \frac{d y^2}{y^2}\right),
 \quad (x,y) \in \mathcal{S}. 
\end{equation}
The martingale problems corresponding to the Laplace-Beltrami operators $\Delta_g$ and $\widetilde \Delta_g$ are
\begin{align*}
f(\overline{X}_t,Y_t)-f(\overline{X}_0,Y_0)- \int_0^t\Delta_gf(\overline{X}_s,Y_s)ds,\quad \textrm{resp.} \quad f(\widetilde{\overline{X}}_t,\widetilde Y_t)-f(\widetilde{\overline{X}}_0,\widetilde Y_0)- \int_0^t\widetilde{\Delta_g}f(\widetilde{\overline{X}}_s,\widetilde Y_s)ds,
\end{align*}
where the operators $\Delta_g$ and $\widetilde{\Delta_g}$  
have (in orthogonal coordinates) the following representation: 
\begin{equation}\label{eq:SABRLaplaceBeltrami}
\begin{array}{rl}
 \Delta_g f 
 =\frac{\beta}{2} y^2 x^{2\beta-1} \frac{\partial f}{\partial x} + Af
 =y^2 \left(\frac{\beta}{2} x^{2\beta-1} \frac{\partial f}{\partial x} + \widetilde{A}f\right)
 =: y^2 \widetilde{\Delta_g} f, \qquad f \in C_0^{\infty}(D).&
\end{array}\end{equation}
\bigskip\\
To formulate the analogous statement to Theorem \ref{Th:TimechangedSABRVolkonskii} in this setting, consider the following system of stochastic differential equations:
\begin{equation}\label{eq:SABRSDEwithDrift}
d\overline{X}_t=Y_t \overline{X}_t^{\beta}dW_t+ \frac{Y_t^2}{2}\beta \overline{X}_t^{2\beta-1}dt,\qquad
dY_t= \alpha Y_t dZ_t, \qquad
d\langle Z,W\rangle_t=\rho dt,
\end{equation}
where $\overline{X}_0=\overline{x}, Y_0=y$, and $t\geq 0.$ Respectively, consider the system
\begin{equation}\label{eq:SABRSDEwithDrift2}
d\widetilde {\overline{X}}_t=\widetilde{\overline{X}}_t^{\beta}dW_t+ \frac{1}{2}\beta \widetilde {\overline{X}}_t^{2\beta-1}dt, \qquad
d\widetilde Y_t= dZ_t, \qquad 
d\langle Z,W\rangle_t=\rho dt,
\end{equation}
where $\widetilde{\overline{X}}_0=\widetilde{\overline{x}}$, $\widetilde Y_0=\widetilde y$, and $t\geq 0$.
\begin{theorem}[Random time change for SABR-Brownian motion]\label{Th:TimechangedSABRwithDrift}
The process $\left(\overline{X}_t,Y_t\right)_{t\geq 0}$ in \eqref{eq:SABRSDEwithDrift} and a time changed version $(\widetilde {\overline{X}}_{\tau^{-1}_t},\wY_{\tau^{-1}_t})_{t\geq 0}$ of the process \eqref{eq:SABRSDEwithDrift2} coincide in law, where the time change is defined as
\begin{align*}
\tau_t=\int_0^t  \widetilde Y_s^{-2}\,ds \quad  \textrm{for} \quad t< T_0^{\widetilde Y}=\inf\{t: \widetilde Y_t =0\}.
\end{align*}
Furthermore, the law of $(\overline{X}_t,Y_t)_{t\geq 0}$ 
solves the martingale problem to $\Delta_g$, and similarly, the law of $(\widetilde {\overline{X}}_{t},\wY_{t})_{t\geq 0}$ solves the martingale problem to $\widetilde \Delta_g$.
\end{theorem}
\section{The semigroup point of view}
Let us turn to applications of the SABR time change. We start with the regularity of the transition semigroup \eqref{eq:Semigroup} of the SABR process \eqref{eq:SABRSDE}. 
For the problem of pricing contingent claims on a forward, suppose that the stochastic process modeling the forward $X = (X_t)$, $t\geq0$ follows SABR-dynamics. 
If $f$ denotes the payoff function of a financial contract, the valuation of the fair price of this contract reduces to the computation of $E^{(x,y)}\left[ f (X_t,Y_t )\right], \ t\in[0,T ]$ under some risk neutral (martingale-) measure, where $(x,y)$ is the initial value of the forward and volatility.
For a suitable set of admissible payoffs $f \in \mathcal{B}$, these expectations form a semigroup 
\begin{equation}\label{eq:Semigroup}
P_t f(x,y)=E^{(x,y)}[f(X_t,Y_t)],\qquad t\geq 0, 
\end{equation}
of bounded linear operators. 
\subsection{Banach Spaces and (generalized) Feller properties}\label{Sec:BanachSpaces}
We speak about Feller properties of the semigroup \eqref{eq:Semigroup} if it satisfies: 
\begin{itemize}
 \item[(F1)] (Semigroup properties) $P_0=Id$, and $P_{t+s}=P_tP_s$ for all $t,s\geq 0$,
 \item[(F2)] (Continuity properties) for all $f \in \mathcal{B}$ and $x\in X$, $\lim_{t\rightarrow0+}P_tf(x)=f(x)$,
 \item[(F3)] (Positive contraction properties) $||P_t||_{L(\mathcal{B})}\leq 1$ for all $t \geq 0$ and \\
$P_t$ is positive for all $t\geq 0$, that is, for any $f\in \mathcal{B}$ $f\geq 0$ implies $P_tf\geq 0$,
\end{itemize}
with choice of admissible payoffs being a suitable Banach space $\mathcal{B}$.
A semigroup with the properties $(F1)-(F3)$ that is invariant on its domain $\mathcal{B}$, i.e.
\begin{equation}\tag{F,\ FD,\ FG}\label{eq:FellerProperties}
 P_t \text{ maps } \mathcal{B} \text{ into itself}
\end{equation}
is referred to as a Feller semigroup (F), a Feller-Dynkin semigroup (FD) or a Generalized-Feller semigroup (FG), depending on the underlying invariant Banach space $\mathcal{B}$. We recall the Feller properties $(F)$ and $(FD)$. 
The transition semigroup $(P_t)_{t\geq 0}$  has the \emph{Feller property} $(F)$, if it acts on the Banach space $\B_{F}=(C_b, ||\cdot||_{\infty})$, that is if
\begin{equation}\tag{F}\label{eq:FellerProp}
P_t\text{ maps } 
C_b=\{f:\R^+\times \R^+\to \R \ |\ f \ \textrm{continuous and bounded}\}
\text{ into itself}. 
\end{equation}
The transition semigroup $(P_t)_{t\geq 0}$  has the \emph{Feller-Dynkin property} $(FD)$, if it acts on the Banach space $\B_{FD}=(C_{\infty}, ||\cdot||_{\infty})$, that is if
\begin{equation}\tag{FD}\label{eq:FellerDyinkinProp}
	 P_t\text{ maps } C_{\infty}=\{f:\R^+\times \R^+\to \R \ |\ f \ \textrm{continuous},\ \lim_{|(x,y)|\rightarrow \infty}f(x,y)=0 \} \text{ into itself}. 
\end{equation}

 Note that strong continuity of the semigroup $P_t$ (which is the required analytic setting in \cite{HansenOsternann}) is an immediate consequence of the Feller-Dynkin property (FD) \cite[p. 369]{kallenberg2002foundations} but not of the weaker Feller property (F). While the Feller property (F) is a direct consequence of the well-posedness of the martingale problem, the Feller-Dynkin property is not always straightforward to verify. Proofs can be found in Appendix \ref{Sec:Proofs}.
\smallskip
\begin{lemma}\label{Th:FellerPropertySABR}
The SABR semigroup \eqref{eq:Semigroup} for \eqref{eq:SABRSDE}
and the heat semigroup \eqref{eq:Semigroup} corresponding to SABR-Brownian motion \eqref{eq:SABRSDEwithDrift} (henceforth SABR-heat semigroup) satisfy 
the Feller property $(F)$.
\end{lemma}
\begin{theorem}\label{Th:FellerDynkinPropertySABR}
The SABR semigroup \eqref{eq:Semigroup} for \eqref{eq:SABRSDE}
and the SABR-heat semigroup both satisfy 
the Feller-Dynkin property (FD). 
\end{theorem}
\subsubsection{The generalized Feller property and weighted spaces}\label{Sec:GeneralizedFeller}
The set of payoffs under consideration---which is $C_{\infty}$---for the Feller-Dynkin property, is rather restrictive. While the European Put is contained, Call contracts are not included. On the other hand, calculating Put option prices in the SABR model may not be straightforward, due to the probability mass at $X=0$, which can accumulate to significant values, see \cite{SABRMassZero}. 
We are therefore interested in Feller-Dynkin-like properties---which have the advantage of implying strong continuity for the semigroup---of the SABR model, which extend the set of admissible payoffs from $C_{\infty}$ to the more realistic case of unbounded function spaces.
In the framework of \cite{SemigroupPOVSplittingSchemes} the set of admissible payoffs is extended to an appropriate Banach space $\cBp$, which includes functions whose growth is controlled by some \emph{admissible} function $\psi$ (see Def. \ref{Def:AdmissibleFunction} below.). This setting is a natural generalization of the Feller-Dynkin property $(FD)$, which is suitable for pricing of options with unbounded payoff. 
Accordingly, we construct here so-called weighted spaces for the SABR model, and prove Feller-like properties (properties (F1)-(F3) and (FG)) on an appropriate invariant Banach space $\B$.
For this, we first recall the framework of \cite{SemigroupPOVSplittingSchemes}. We denote by $D$ the state space of the stochastic process under consideration.
\begin{definition}[Admissible weight functions and weighted spaces]\label{Def:AdmissibleFunction}
On a completely regular Hausdorff space $D$, a function $\psi:D\rightarrow(0,\infty)$ is an \emph{admissible
weight function} if the sub-level sets
\begin{equation}\label{eq:CompactSublevelSets}
K_R:=\{x \in D: \psi(x)\leq R\}
\end{equation}
are compact for all $R>0$. A pair $(D,\psi)$ where $D$ is a completely regular Hausdorff space and $\psi$ an admissible weight function is called a \emph{weighted space}.
\end{definition}
Note that admissibility renders weight functions $\psi$ lower semi-continuous and bounded from below by some $\delta>0$, cf. \cite{SemigroupPOVSplittingSchemes}. 

\begin{lemma}
Consider the set
\begin{equation*}\label{eq:BPsiSpace}
 B_{\psi}(D):=\left\{f:D\rightarrow \R: \sup_{x\in D} \psi(x)^{-1} |f(x)|<\infty \right\},\quad \textrm{with}\quad ||f||_{\psi}:=\sup_{x\in D}\psi(x)^{-1} ||f(x)||_{\R}.
\end{equation*}
Then the pair  $(B_{\psi}(D),||\cdot||_{\psi})$ is a Banach space and $C_b(D)\subset B_{\psi}(D)$, where $(C_b(D),||\cdot||_{\infty})$ denotes the space 
of bounded continuous functions endowed with the supremum norm. \\
See \cite[Section 2.1]{SemigroupPOVSplittingSchemes}.
\end{lemma}
\begin{definition}
\label{Def:WeightedSpaces}
For a weighted space $(D,\psi)$ we define the function space $\cBp$ as the closure 
\begin{equation}
\cBp:=\overline{C_b(D)}^{||\cdot||_{\psi}} 
\end{equation}
of the set of bounded continuous functions 
in $B_{\psi}(D)$ under the norm $||\cdot||_{\psi}$. We refer to elements of the space $\cBp$ as \emph{functions with growth controlled by $\psi$}. 
\end{definition}
The spaces $\{\cBp$: $\psi$ admissible$\}$ in Definition \ref{Def:WeightedSpaces} above, coincide with the spaces constructed in \cite[equation (2.2)]{RoecknerSobol}, where the authors study well-posedness of martingale problems in the sense of Stroock and Varadhan in an SPDE setting.
A key feature of $\cBp$ spaces is that they allow for complete characterization of their respective dual spaces $\cBp(D)^{*}$ via a Riesz representation, which was proved for $\cBp$-spaces in \cite[Theorem 2.5]{SemigroupPOVSplittingSchemes} and \cite[Theorem 5.1]{RoecknerSobol} independently. Also, these spaces are often separable, which is an important property in all questions of implementable approximations.
We stated above that the considered Banach spaces are a generalization of the space $C_{\infty}(D)$
to ones which include functions of unbounded growth.
In fact, the following criterion given in \cite{SemigroupPOVSplittingSchemes} restores the \emph{vanishing-at-infinity} property for functions $\widetilde{f}:=\frac{|f|}{\psi}$ \emph{weighted by admissible functions $\psi$}:
\begin{proposition}\label{Th:WeightedVanishingInfty}
For any $f:D\longrightarrow \R$ it holds that  $f \in \mathcal{B}^{\psi}$ if and only if $f|_{K_R} \in C(K_R)$ for all $R>0$ and
\begin{equation*}
\lim_{R\rightarrow \infty}\sup_{x \in D  \setminus K_R}\psi(x)^{-1}|f(x)|=0,
\end{equation*}
where $K_R:=\{x \in D: \psi(x)\leq R\}$ for positive real $R$. 
\end{proposition}
\begin{proof}
See \cite[Theorem 2.7]{SemigroupPOVSplittingSchemes} 
\end{proof}

\begin{definition}[Generalized Feller Property]\label{Def:GeneralizedFellerProp}
Consider the family $(P_t)_{t\geq 0}$ of bounded linear operators on a weighted space $\cBp(D)$. The semigroup has the \emph{generalized Feller property} (cf. \cite[Section 3]{SemigroupPOVSplittingSchemes}) if ($F1$), ($F2$) and the following properties are satisfied:
\begin{itemize} 
 \item[($\widetilde F3$)] if $f\in \cBp(X)$, and $f\geq 0$ then $P_tf\geq 0$ for all $t\geq0$. 
 \item[($\widetilde F4$)] there exist a constant $C\in \R$ and $\epsilon >0$, such that
$||P_t||_{L(\cBp(D))}\leq C$ for all $t \in [0,\epsilon]$.
\end{itemize}
\end{definition}
It is immediate that ($\widetilde F3$) covers the positivity statement of (F3).
The crucial property is ($\widetilde F4$), which yields the contraction property of (F3) such as the invariance (FG) of the domain $\cBp$ under the semigroup action. 
As stated above, the applicability of convergence theorems requires an analytic setting in which the SABR semigroup is a strongly continuous contraction semigroup on a suitable Banach space $\B$. The strong continuity of $P_t$
on $\cBp$ is stated in the following theorem.
\begin{theorem}\label{Th:StrongContinuityBpsi}
Let $(P_t)_{t\geq0}$ be a generalized Feller semigroup on the Banach space $\cBp(D)$. Then
$(P_t)_{t\geq0}$ is a strongly continuous semigroup on $(\cBp,||\cdot||_{\psi})$.
\end{theorem}
\begin{proof}
See \cite[Theorem 3.2]{SemigroupPOVSplittingSchemes}. 
\end{proof}

The following theorem and corollary provide a characterization of suitable Banach spaces $\cBp$ for the SABR processes \eqref{eq:SABRSDE} and \eqref{eq:SABRSDEwithDrift} which contain payoff functions of polynomial growth, and on which corresponding semigroups are generalized Feller semigroups (and hence strongly continuous by Theorem \ref{Th:StrongContinuityBpsi}):
\begin{theorem}[Generalized Feller properties for the SABR-heat and SABR semigroups]\label{Th:GeneralizedFellerPropSABRwithDrift}
Consider the SABR-heat semigroup (resp. the SABR semigroup)
\begin{equation}\label{eq:HeatSemigrouponBpsi}
 P_tf(x,y):=\mathbb{E}^{(x,y)}\left[f(X_t,Y_t)\right], \qquad t\geq0,\ (x,y)\in D,\ f \in \mathcal{B}^{\psi_{c,n}}(D),
\end{equation} 
where $(X,Y)$ is the process \eqref{eq:SABRSDEwithDrift} (resp. the process \eqref{eq:SABRSDE}).
Consider furthermore a family of functions $\psi_{c,n}:=L_{n(n+1)}\circ r_c, \ n\in \mathbb{N}$, where $L_{n(n+1)}$ denote the Legendre polynomials of order $n,\ n\in \mathbb{N}$, and $r_c:(0,\infty)^2\rightarrow (0,\infty)$ denote the family of functions
\begin{equation}\label{eq:Cosh}
r_c(x,y):= \frac{1+y^2}{2 y}+ 
\frac{ \Big( \frac{x^{1-\beta}}{1-\beta} - \rho y-c  \Big)^2}
{(1 - \rho^2)2 y}, \qquad (x,y) \in (0,\infty)^2
\end{equation}
for a $c\in [0,\infty)$. Then the following statements hold:
\begin{itemize}
 \item[(i)]
The functions $\psi_{c,n}, \ n\in \mathbb{N}$, $c\in[0,\infty)$ are admissible weight functions (cf. Definition \ref{Def:AdmissibleFunction}). 
 \item[(ii)] The SABR-heat semigroup is a generalized Feller semigroup on the Banach spaces $\mathcal{B}^{\psi_{c,n}}(D)$, for any $(c,n)\in [1,\infty)\times \N$ and any configuration of SABR parameters $(\rho,\beta) \in(-1,1)\times[0,1)$. The same statement holds for $(c,n)\in [0,1)\times \N$ in the parameter regime $(\rho,\beta)\in (0,1)\times[0,1)$ and as long as $c>|\rho|$ also for the parameters $(\rho,\beta)\in (-1,0]\times[0,1)$. 
In the case $c\leq |\rho|$ with $\rho<0$, the statement holds under the restriction 
\begin{align*}(\rho,\beta) \in (-1,0)\times \left\{\frac{2m-1}{2m}, m\in \mathbb{N}\right\}.
\end{align*}
\item[(iii)] The function $\psi: (0,\infty)^2  \rightarrow (0,\infty)$ defined for any $\beta\in [0,1)$ as
\begin{equation}\label{eq:AdhocAdmissibleWeight}
\psi(x,y):=y+2x^{1-\beta}+\frac{x^{2-2\beta}}{y}, \qquad (x,y) \in (0,\infty)^2 
\end{equation}
is an admissible weight function and for any $(\beta,\rho) \in [0,1)\times(-1,1)$ the SABR semigroup is a generalized Feller semigroup---and hence strongly continuous---on the space $\mathcal{B}^{\psi}(D)$, where $\psi$ is the function \eqref{eq:AdhocAdmissibleWeight} with $\beta$ matching the SABR parameter.
\item[(iv)] Furthermore, the SABR semigroup is a generalized Feller semigroup on the spaces $\mathcal{B}^{\psi_{0,n}}(D)$, $n\in \mathbb{N}$, for any configuration of SABR parameters in 
\begin{align*}
(\rho,\beta) \in (-1,0)\times\{0\}\cup \left\{\frac{2m-1}{2m}, m\in \mathbb{N}\right\}.
\end{align*}
\end{itemize}
\end{theorem}
\begin{corollary}\label{Th:Lineargrowth}
For any $\beta\in [0,1)$ there is an $N\in \mathbb{N}$, such that the SABR (resp. the SABR-heat) semigroup is a strongly continuous contraction semigroup on a space $\mathcal{B}^{\psi_n}(D)$ which contains payoff functions of European call contracts. 
\end{corollary}
For example if $\rho<0$ and $\beta=\frac{1}{2}$, then the space $\mathcal{B}^{\psi}(D)$ with $\psi$ as in \eqref{eq:AdhocAdmissibleWeight} as well as any of the spaces $\mathcal{B}^{\psi_{c,n}}(D)$ with $c=0$ and $n\geq1$ is suitable. 
For $\beta=\frac{3}{4}$, one can choose any $\mathcal{B}^{\psi_{c,n}}(D)$ with $c=0$ and $n\geq2$.
\subsection{Hilbert Spaces and Dirichlet forms}\label{Sec:HilbertSpacesDirichlet}
In the previous section we investigated regularity (strong continuity) properties of the semigroups \eqref{eq:Semigroup} corresponding to SABR-type processes on the Banach spaces $\mathcal{B}_{F}~=~(C_b,||\cdot||_{\infty})$ and $\mathcal{B}_{FD}~=~(C_{\infty},||\cdot||_{\infty})$ such as $\mathcal{B}_{FG}~=~(\cBp,||\cdot||_{\psi})$.
In the present section, the underlying spaces are Hilbert spaces $(\Hh,||\cdot||_{\Hh})$, endowed with a norm and scalar product $||\cdot||_{\Hh}=(\cdot,\cdot)_{\Hh}^{1/2}$. The Riesz representation property (which was crucial on Banach spaces constructed in Section \ref{Sec:BanachSpaces}) is here naturally encoded in the Hilbert space structure, and so are symmetry properties of the semigroup and of corresponding closed forms. The latter are essential in order to establish Dirichlet forms for SABR in order to characterize its short-time asymptotic behavior near zero. We recall here some necessary concepts and refer the reader to \cite{BouleauDenis,BouleauHirsch,FukushimaOshimaTakeda,TerElstDiffusion} for full details.
\begin{definition}\label{Def:SymmetricSemigroup}
A family $(P_t)_{t\geq0}$ of operators with domain $\Hh$
satisfying
\begin{itemize}
 \item [(S1)] (Semigroup properties) $P_0=Id$ and for all $s,t>0$ $P_{t+s}=P_tP_s$
 \item [(S2)] (Strong continuity) for all $f \in \Hh$  $\lim_{t\rightarrow0}||P_tf-f||_{\Hh}=0$ 
 \item [(S3)] (Contraction property) for all $f \in \Hh$ and all $t>0$ $||P_tf||_{\Hh}\leq||f||_{\Hh}$, 
\end{itemize}
is a \emph{strongly continuous symmetric semigroup} on $(\Hh,(\cdot,\cdot)_{\Hh})$ if properties $(S1)-(S3)$ hold, and
\begin{equation}\tag{S}\label{eq:FellerPropertiesSymmetricSemigroup}
\begin{array}{rl}
(P_tf,g)_{\Hh}=(f,P_tg)_{\Hh} &\textrm{for all }f,g\in \Hh \quad \textrm{and}\\
P_t\text{ maps } \Hh \text{ into itself}.&
\end{array}\end{equation}
\end{definition}
The set of symmetric semigroups will be denoted by $\bPi$.
Furthermore, the set\footnote{The correspondence between the sets $\bPi$ and $\bF$ being the one described in \cite[Section 1.3]{FukushimaOshimaTakeda}, for example via Riesz representation for the generator of the semigroup.} of closed (symmetric bilinear) forms on $\mathcal{H}$ is denoted by $\bF$ and is defined as follows, cf. \cite[Section 1.3]{FukushimaOshimaTakeda}.
\begin{definition}
\label{Def:ClosedForm} 
A \emph{symmetric form} on $(\Hh,(\cdot,\cdot)_{\Hh})$ is a pair $(\cE ,D(\cE))$, where $\cE$ is a non-negative definite symmetric bilinear form with dense domain $D(\cE)\subset \Hh$. That is, $\cE$ is a symmetric form if $D(\cE)\subset \Hh$ dense linear subspace, and satisfies
\begin{enumerate}
\item[(B1)] (Nonnegativity) $\cE:D(\cE)\times D(\cE) \longrightarrow \R$ and $\cE(f,f)\geq0$ for all $f \in D(\cE)$
\item[(B2)] (Symmetry) $\cE(f,g)=\cE(g,f)$ and $\cE(f,f)\geq0$ for all $f,g \in D(\cE)$
\item[(B3)] (Bilinearity) $\cE(\alpha f,g)=\alpha \cE(f,g)$, such as $\cE(f_1+f_2,g)=\cE(f_1,g)+\cE(f_2,g)$ \\
\phantom{(Bilinearity) }for all $f,g,f_1,f_2 \in D(\cE)$ and $\alpha \in \R$.
\end{enumerate}
A \emph{closed form} is a symmetric form on $(\Hh,(\cdot,\cdot)_{\Hh})$ 
such that the pair $\left(D(\mathcal{E}), ||\cdot||_{\mathcal{E}_1}\right)$
with norm $||f||_{\mathcal{E}_1}:=\sqrt{\cE(f,f)+||f||_{\Hh}^2}$, $f \in
D(\mathcal{E})$ is a Hilbert space.
That is, if 
\begin{enumerate}
\item[(B4)] (Completeness)
the space $D(\mathcal{E})$ is complete with respect to the norm 
$||\cdot||_{\mathcal{E}_1}$.
\end{enumerate}
A \emph{Dirichlet form} is a closed form on $(\Hh,(\cdot,\cdot)_{\Hh})$  which satisfies
\begin{enumerate}
 \item[(B5)] (Markovianity) for all $u\in \DcE$, $u^+\wedge 1 \in \DcE$ and $\cE(u+u^+\wedge 1,u^+\wedge 1)\leq \cE(u,u).$ 
\end{enumerate}
\begin{definition}\label{Def:SelfAdjointOperator}
Let $D(A)\subset \Hh$ be a dense subset of the Hilbert space $\Hh$. A linear operator 
$$\textrm{$A:D(A)\rightarrow \Hh$ 
 is self-adjoint if
$A=A^*$ and $D(A)=D(A^*)$, }$$
where $A^*$ denotes the adjoint operator:
$(Af,g)_{\Hh}=(f,A^*g)_{\Hh}$ for all $f\in D(A)$, $g\in D(A^*)$, and 
\begin{equation}\label{eq:DomainAdjointOperator}
D(A^*)=\{g\in \Hh: \ \exists v \in \Hh \quad\textrm{s.th.} \quad \forall f \in D(A) \quad (Af,g)_{\Hh} =(f,v)_{\Hh}\}. 
\end{equation}
A self adjoint operator $A$ is called negative if $\textrm{spec}(A)\subset(-\infty,0]$. The set of negative self-adjoint operators is denoted by $\bA$.
\end{definition}
\begin{proposition}\label{Th:OnetoOneCorrespFormSemigroupOperator} Let $\bF,\ \bA$ and $\bPi$ be the sets of closed forms,
of negative self-adjoint operators and of symmetric semigroups on the Hilbert space $\Hh$ 
as in Definitions \ref{Def:ClosedForm}, \ref{Def:SelfAdjointOperator}
and \ref{Def:SymmetricSemigroup} respectively. Then the following diagram is commmutative and the maps $\Phi,\Theta$ and $\Psi$ are bijective:
\begin{align}\label{eq:FormsSemigroupsOperatorsmap}
 \xymatrix{\bF\ar[r]^{\Phi} &\bA\\
          \bPi\ar[u]^{\Psi}\ar[ur]_{\Theta}&}. 
\end{align}
\end{proposition}
See \cite{BouleauHirsch} for details and for the construction of the maps \eqref{eq:FormsSemigroupsOperatorsmap}. 
If $(\cE,\DcE)$ is a Dirichlet form, we call $(\Phi(\cE),\Phi(\DcE))$ the \emph{generator} of the Dirichlet form. Below we construct Dirichlet forms---see the Theorem \ref{Th:SABRDirichletForm}---whose generators correspond to the infinitesimal generators \eqref{eq:GeneratorTimeChangedSABR} and \eqref{eq:SABRLaplaceBeltrami} of the SABR model and of the SABR-Brownian motion on a suitable domain $\Phi(\DcE)$.
\end{definition}

\subsubsection{Symmetric Dirichlet forms for SABR-type processes}\label{Sec:DichichletformsforSABR}
Consider the domain $S:=\R\times (0,\infty)$ with the measures
\begin{align}
m_j(x,y)dxdy:=\frac{1}{\rrho x^{\beta(1+j)}y^2}dxdy, \quad \textrm{for} \quad \textrm{j=0,1}.
\end{align}
Furthermore, let $\Hh_j:=\Ll^2(S, m_j(x,y)dxdy)$, $j=0,1$ denote weighted spaces with measures $m_j(x,y)dxdy$, $j=0,1$ as above.
On $\Hh_j$, $j=0,1$ we consider the bilinear forms
\begin{align}\label{eq:DirichletformSABRwithdrift1}
\mathcal{E}_{m_j}(f_1,f_2):= 
 \frac{1}{2}\int_{S} \Gamma(f_1,f_2) m_j(x,y)dxdy,
\end{align}
where the integrand is defined as
\begin{equation}\label{eq:Energymeasure}
\Gamma(f_1,f_2):= y^2\left(x^{2\beta} \partial_{x}f_1\partial_{x}f_2+2 \rho x^{\beta} \partial_{x}f_1\partial_{y}f_2+\partial_{y}f_1\partial_{y}f_2\right).
\end{equation}
\begin{theorem}[Symmetric Dirichlet forms for SABR-Brownian motion and uncorrelated SABR]\label{Th:SABRDirichletForm}
The spaces $\Hh_j$ are Hilbert spaces for $j=0,1$ and the following statements hold: 
\begin{itemize}
\item[(i)] The pair $(\cE_{m_j}, C^{\infty}_0(S))$ is a symmetric form on $\Hh_j$ for $j=0,1$ and \emph{closable}\footnote{See sections \ref{Sec:Closability}, and \ref{Sec:BeurlingDenyLeJan} for the definition and implications of closability.} for all $\beta\in [0,1)$ whenever  $j=0$ and for all $\beta\in[0,1/2)$ whenever $j=1$.
\item[(ii)] The pair $(\cE_{m_j}, D(\cE_{m_j}))$ is a Dirichlet form on $\Hh_j$ for $j=0$ for all $\beta\in [0,1)$ and for $j=1$ for all $\beta\in[0,1/2)$ with the domain
\begin{equation}\label{eq:DirichletformDomain}\begin{array}{ll}
 D(\cE_{m_j}):=\Big\{&u \in \Hh_j, \ \mathcal{B}(S)-\textrm{measurable},\ such\ that \ \forall \ u_{\bar y},u_{\bar x}\ \exists \ \textrm{abs. cont. version},\\ 
& \textrm{and such that}   \ \Gamma(u,u) \in \Ll^1(S,m_j(x,y)dxdy) \quad  \Big\}
\end{array}
\end{equation}
where $u_{\bar y}:=(x\mapsto u(x,\bar{y})), u_{\bar x}:=(y\mapsto u(\bar{x},y))$.
\item[(iii)] The generator $\Phi(\cE_{m_0})$ satisfies 
\begin{equation*}
(\Phi(\cE_{m_0}) f_1,f_2)_{\mathcal{H}_{m_0}}=(\Delta_g f_1,f_2)_{\mathcal{H}_{m_0}}=\mathcal{E}_{m_0}(f_1,f_2),\quad \textrm{for} \quad f_1, f_2 \in 
C_0^{\infty}(\Ss),
\end{equation*}
where $\Delta_g$ is the Laplace-Beltrami operator \eqref{eq:SABRLaplaceBeltrami}.
In the particular case $\beta=0$ it holds that
\begin{equation}
(\Phi(\cE_{m_0}) f_1,f_2)_{\mathcal{H}_{m_0}}=(A f_1,f_2)_{\mathcal{H}_{m_0}}=\mathcal{E}_{m_0}(f_1,f_2),\quad f_1, f_2 \in 
C_0^{\infty}(\Ss),
\end{equation}
where $A$ denotes the SABR infinitesimal generator \eqref{eq:GeneratorTimeChangedSABR}.
\item[(iv)] Furthermore, in the case $\rho=0$ it holds that 
\begin{equation*}
(\Phi(\cE_{m_1}) f_1,f_2)_{\mathcal{H}_{m_1}}=(A f_1,f_2)_{\mathcal{H}_{m_1}}=\mathcal{E}_{m_1}(f_1,f_2),\quad \textrm{for} \quad f_1, f_2 \in
C_0^{\infty}(\Ss).\\
\end{equation*}
\end{itemize}
\end{theorem}
Analogous statements can be formulated in the the one-dimensional (CEV) situation. Here, we consider the weighted space 
$$H_{\beta}:=L^2(\R, m_{\beta}(x) dx),$$ where $m_{\beta}(x)\equiv 1/x^{2\beta}$. On $H_{\beta}$ we define the following bilinear form:
\begin{align}\label{eq:DirichletformSABRwithdrift3}
\mathcal{E}_c(f_1,f_2):=
 \frac{1}{2}\int_{\R} \Gamma^{c}(f_1,f_2) m_{\beta}(x)dx,
\end{align}
where the integrand is defined as
\begin{equation}\label{eq:EnergymeasureCEV}
\Gamma^{c}(f_1,f_2):= \sigma x^{2\beta} \partial_{x}f_1\partial_{x}f_2.
\end{equation}
\begin{lemma}[Symmetric Dirichlet form for CEV]\label{Th:CEVSymmDirichletform}
The following statements hold:
\begin{itemize}
\item[(i)] The pair $(\cE_{c}, C^{\infty}_0(\R))$ is a symmetric form on $H_{\beta}$ for $\beta\in[0,1]$ and \emph{closable} for $\beta\in[0,\frac{1}{2})$.
\item[(ii)] The pair $(\cE_{c}, D(\cE_{c}))$ is a Dirichlet form on $H_{\beta}$ for $\beta\in[0,\frac{1}{2})$ with the domain
\begin{equation}\label{eq:DirichletformDomain}\begin{array}{ll}
 D(\cE_{c}):=\Big\{&u \in \Hh_j, \ \mathcal{B}(\R)-\textrm{measurable},\ such \ that \  \exists \ \textrm{abs. cont. version}\\ 
& \textrm{for which }   \ \Gamma^c(u,u) \in L^1\left(\R,m(x,y)dxdy\right) \quad  \Big\}
\end{array}
\end{equation}
\item[(iii)] It holds for all $\beta\in[0,1]$ that 
$(\sigma^2x^{2\beta}\partial_{xx} f_1,f_2)_{H}=\widetilde{\mathcal{E}}_{c}(f_1,f_2),\ \textrm{for} \ f_1, f_2 \in C_0^{\infty}((0,\infty)).$
\end{itemize}
\end{lemma}

\begin{remark}\label{Rem:CEV}
Note that in \cite{ReichmannComputationalMethods} further (non-symmetric) bilinear forms for the CEV model are considered on the larger spaces spaces $L^2((0,\infty),1/x^{\beta} dx)$ and $L^2((0,\infty),dx)$.
Note also, that for $\beta\in [0,1]$, the CEV-manifold is $((0,\infty),g_c)$, with the Riemannian the metric $g_{c}(x,x):=\sigma^2x^{2\beta}dx\otimes dx$. In particular, 
the Riemannian distance between points $a,b>0$ remains finite as $a\rightarrow 0$ for all $\beta\in [0,1)$, but the limit becomes infinite for $\beta=1$. 
Note also, that although the weights $m_{\beta}$ for any $\beta$ ensure symmetry of the bilinear form corresponding to the CEV generator for any $\beta\in[0,1]$, but for $\beta\in[1/2,1]$ the measure $m_{\beta}$ is no longer a Radon measure.
\end{remark}

\begin{theorem}[Dirichlet forms for SABR: Possible parameter configurations]\label{Th:SABRDirichletFormSymmetricNo}
The only possible parameter configurations of the SABR model \eqref{eq:SABRSDE} for which there exists a weighted space $\Hh_{m}:=\Ll^2\left(\Ss, m(x,y) dxdy\right)$ 
for a $dxdy$-a.s. positive Borel function $m:\Ss\rightarrow [0,\infty)$,  
and a bilinear form
\begin{equation*}
\mathcal{E}_{m}(f_1,f_2):=
 \frac{1}{2}\int_{S} \Gamma(f_1,f_2) m(x,y)dxdy,
\end{equation*}
for a symmetric operator $\Gamma$, 
which satisfies for the SABR generator\footnote{Defined in \eqref{eq:GeneratorTimeChangedSABR}.} the relation
\begin{equation*}
(A f_1,f_2)_{\mathcal{H}_{m}}=\mathcal{E}_{m}(f_1,f_2),\quad \textrm{for} \quad f_1, f_2 \in C_0^{\infty}(\Ss)
\end{equation*}
are the following cases:
\begin{itemize}
 \item[(i)] $\beta=0$, $\rho\in(0,1)$, $\nu>0$: For these parameters  
it holds that $Af=\Delta_gf$ for all $C_0^{\infty}(\Ss)$. Note in particular, that in the special case $\rho=0=\beta$, $\Delta_g$ is the Laplace-Beltrami operator of the hyperbolic plane, cf. \cite{HLW}.
\item[(ii)] $\beta\in[0,1]$, $\rho=0$, $\nu>0$: For this parameter configuration it holds that $Af=\Delta_{\Upsilon \mu}f$ for all $C_0^{\infty}(\Ss)$, where $\Delta_{\Upsilon \mu}$ denotes the Laplace-Beltrami operator of a weighted manifold\footnote{See \eqref{eq:LaplBeltramiWeightedManifold} and \cite[Definition 3.17]{grigoryan09AnalysisManifolds} for full details.} $(\Ss,g,\Upsilon \mu)$, where $g$ is the corresponding Riemannian measure (recall \eqref{eq:ManifoldMetric}), $\mu=\sqrt{\det(g)}$ denotes the associated Riemannian volume form, and $\Upsilon(x,y)=x^{-\beta}$ the weight function.
\item[(iii)]  $\nu=0$ and $\beta \in [0,1]$: In this (univariate) case the model \eqref{eq:SABRSDE} reduces to the CEV model.
\item[(iv)] $\beta=1$, $\rho\in(0,1)$, $\nu>0$: In this case the measure which allows us to pass from the SABR generator to a symmetric bilinear form reads 
\begin{equation}\label{eq:SpeedMeasureBeta1} 
m(x,y)=\frac{1}{y^2 x^{1+1/\rrho^2}} \exp\left(\tfrac{\rho}{\rrho}y\right).
\end{equation}
\end{itemize}
For all other parameter configurations of the SABR model \eqref{eq:SABRSDE} the symmetry property of the associated Bilinear form $\mathcal{E}_{m}$ breaks down for any positive Borel function $m(x,y)$.
\end{theorem}
\begin{remark}
Statement (i) of Theorem \ref{Th:SABRDirichletFormSymmetricNo} is covered in (iii) of Theorem \ref{Th:SABRDirichletForm}.
Furthermore, statement (ii) of Theorem \ref{Th:SABRDirichletFormSymmetricNo} is covered in (iv) of Theorem \ref{Th:SABRDirichletForm}. Finally, statement (iii) of Theorem \ref{Th:SABRDirichletFormSymmetricNo} is covered in Lemma \ref{Th:CEVSymmDirichletform}. The crucial statement in Theorem \ref{Th:SABRDirichletFormSymmetricNo} is statement (v), that (i)-(iv) are in fact all possible parameter configurations. Proofs can be found in Appendix \ref{Sec:Proofs}.
\end{remark}
\subsubsection{Dirichlet forms for the time changed processes, and stochastic representation}For the time-changed processes in \eqref{eq:SABRdynamicsTimechanged} and \eqref{eq:SABRSDEwithDrift} analogous statements to Theorem \ref{Th:SABRDirichletForm} hold:\\
Let $\widetilde \Hh_j$ for $j=0,1$ denote, as in Theorem \ref{Th:SABRDirichletForm} above, weighted $\mathcal{L}^2$-spaces with weighted measures $$\widetilde m_j(x,y)dxdy:=\frac{1}{\rrho x^{\beta(1+j)}}dxdy.$$ On $\widetilde \Hh_j$, $j=0,1$ consider the following bilinear forms:
\begin{align}\label{eq:DirichletformSABRwithdrift2}
\widetilde{\mathcal{E}}_{m_j}(f_1,f_2):= 
 \frac{1}{2}\int_{\R^2} \widetilde \Gamma(f_1,f_2) \widetilde m_j(x,y)dxdy,
\end{align}
where the integrand is defined as
\begin{equation}\label{eq:Energymeasure2}
\widetilde\Gamma(f_1,f_2):=x^{2\beta} \partial_{x}f_1\partial_{x}f_2+2 \rho x^{\beta} \partial_{x}f_1\partial_{y}f_2+\partial_{y}f_1\partial_{y}f_2.
\end{equation}
\begin{theorem}[Symmetric Dirichlet forms for the time-changed processes]\label{Th:SABRDirichletFormTimechanged}On the spaces $\widetilde \Hh_j$, $j=0,1$ the following statements hold:
\begin{itemize}
\item[(i)] The pair $(\widetilde{\cE}_{m_j}, C^{\infty}_0(\R^2))$ is a symmetric form on $\widetilde \Hh_j$ for $j=0,1$ and closable for all $\beta\in [0,1)$ whenever $j=0$ and for all $\beta\in[0,1/2)$ whenever $j=1$.
\item[(ii)] The pair $(\widetilde{\cE}_{m_j}, D(\widetilde{\cE}_{m_j}))$ is a Dirichlet form on $\widetilde \Hh_j$ for all $\beta\in [0,1)$ whenever $j=0$ and for all $\beta\in[0,1/2)$ whenever $j=1$. Its domain is
\begin{equation}\label{eq:DirichletformDomain}\begin{array}{ll}
 D(\widetilde{\cE}_{m_j}):=\Big\{&u \in \widetilde{\Hh}_j, \ \mathcal{B}(\R^2)-\textrm{mb. },\ such \ that \ \forall \ u_{\bar y},u_{\bar x}\ \exists \ \textrm{abs. cont. version},\\ 
& \textrm{and such that }   \ \Gamma(u,u) \in \Ll^1(S,m_j(x,y)dxdy) \quad  \Big\}
\end{array}
\end{equation}
where $u_{\bar y}:=(x\mapsto u(x,\bar{y})), u_{\bar x}:=(y\mapsto u(\bar{x},y))$.
\item[(iii)] The generator $\Phi(\widetilde{\cE}_{m_0})$ satisfies 
\begin{equation*}
(\Phi(\widetilde{\cE}_{m_0}) f_1,f_2)_{\widetilde{\mathcal{H}}_{m_0}}=(\widetilde \Delta_g f_1,f_2)_{\widetilde{\mathcal{H}}_{m_0}}=\widetilde{\mathcal{E}}_{m_0}(f_1,f_2),\quad \textrm{for} \quad f_1, f_2 \in 
C_0^{\infty}(\Ss),
\end{equation*}
where $\widetilde \Delta_g$ is as in \eqref{eq:SABRLaplaceBeltrami}. 
For $\beta=0$  it holds in particular that
\begin{equation}
(\Phi(\widetilde{\cE}_{m_0}) f_1,f_2)_{\mathcal{H}_{m_0}}=(\widetilde A f_1,f_2)_{\widetilde{\mathcal{H}}_{m_0}}=\widetilde{\mathcal{E}}_{m_0}(f_1,f_2),\quad f_1, f_2 \in 
C_0^{\infty}(\Ss),
\end{equation}
where $\widetilde A$ is as in \eqref{eq:GeneratorTimeChangedSABR}.
\item[(iv)] Furthermore, $\widetilde A$ in \eqref{eq:GeneratorTimeChangedSABR} for $\rho=0$ and the generator $\Phi(\widetilde{\cE}_{m_1})$ satisfy 
\begin{equation*}
(\Phi(\widetilde{\cE}_{m_1}) f_1,f_2)_{\widetilde{\mathcal{H}}_{m_1}}=(\widetilde A f_1,f_2)_{\widetilde{\mathcal{H}}_{m_1}}=\widetilde{\mathcal{E}}_{m_1}(f_1,f_2),\quad \textrm{for} \quad f_1, f_2 \in 
C_0^{\infty}(\Ss).
\end{equation*}
\end{itemize}
\end{theorem}

\begin{remark}[Stochastic representation]Consider the following system of stochastic differential equations:
\begin{equation}\label{eq:SABRDirichletprocess}
dX_t=Y_tX^{\beta}_tdW_t, \qquad
 dY_t=Y_tdZ_t - \rho \beta Y_t^2  X_t^{\beta-1}dt, \qquad
 d\langle W, Z\rangle_t=\rho,
\end{equation}
and respectively consider the system
\begin{equation}\label{eq:SABRDirichletprocess2}
 d \widetilde X_t=\widetilde X^{\beta}_tdW_t, \qquad
 d \widetilde Y_t=dZ_t - \frac{\rho}{2}\beta  \widetilde X_t^{\beta-1}dt, \qquad
 d\langle W, Z\rangle_t=\rho.
\end{equation}
The infinitesimal generators corresponding to \eqref{eq:SABRDirichletprocess} and \eqref{eq:SABRDirichletprocess2} coincide on the domain $C_0^{\infty}(\Ss)$ with the generators of the Dirichlet forms $\Phi(\cE_{m_1})$, resp. $\Phi(\widetilde{\cE}_{m_1})$ in $(iv)$ of Theorem \ref{Th:SABRDirichletForm} resp. Theorem \ref{Th:SABRDirichletFormTimechanged} for any $\beta\in [0,1]$ and $\rho\in (0,1)$. Note that for $\rho=0$ the system \eqref{eq:SABRDirichletprocess}, coincides the system for uncorrelated SABR model \eqref{eq:SABRSDE}. Analogous statements hold for the systems \eqref{eq:SABRDirichletprocess2} and \eqref{eq:SABRdynamicsTimechanged}.
\end{remark}

\section{Asymptotics}\label{Sec:Asymptotics} 
Another direct application of the SABR time change \eqref{eq:Thetimechange} is that it enables us to characterise the large time behaviour the SABR process, more precisely, the distribution of $X_t$ as $t\rightarrow \infty$. 
In \cite[Section 4]{Hobson} similar asymptotic conclusions are derived in a log-linear setting and in \cite[Example 5.2]{Hobson} a special case of the SABR model is presented ($\beta=0,\rho=1$), where the process a.s. has a non-trivial limit. 
A characterization of the large-time behaviour of the SABR process is of interest beyond this special case. Therefore we highlight here that this characterization can be easily extended to general ($\beta\in [0,1]$) for the uncorrelated SABR model and for SABR-Brownian motion, and outline the proof in Appendix \ref{Sec:Proofs}. 
\begin{figure}[h!]
\begin{center}
\includegraphics[scale=0.45]{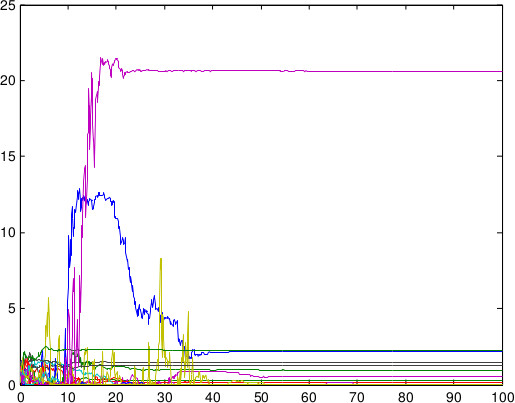}
\caption{Large-time behavior of the SABR process on the time horizon $T=100$ years.
The simulation was done by an explicit Euler scheme with absorbing boundary conditions at zero and shows $N=10$ sample paths of the $X$-coordinate of the SABR process for the parameters $\nu=1, \rho=0.9, \beta=0.5$. The sample paths are either absorbed at zero or they level off at a non-trivial (positive) value.}
\label{fig:Mass0InfTimeX0}
\end{center}
\end{figure}

\subsection{Large-time asymptotics}\label{Sec:AsymptoticsLargetime}
The SABR process and the SABR-Brownian motion have a non-trivial large-time behavior and the time change gives insight into the sample-path behaviour of the model. The second coordinate process $(Y_t)_{t\geq 0}$ of \eqref{eq:SABRSDE} (resp. of \eqref{eq:SABRSDEwithDrift}) is a driftless geometric Brownian motion and as such converges almost surely to $0$. For the first coordinate two scenarios are possible: either the geometric Brownian motion $Y$ stays long enough over some threshold so that the first coordinate process $\wX$ (resp. $\overline{X}$) ``has time'' to hit zero; or, $Y$ gets small quickly enough so that the fluctuations of $\wX$ level off and $\wX$ (resp. $\overline{X}$) converges to a non-zero limit. The next theorem shows that both happen with positive probability both for the (uncorrelated) SABR model \eqref{eq:SABRSDE} and for the SABR-Brownian motion \eqref{eq:SABRSDEwithDrift}: 
\begin{theorem}\label{Th:AsymptoticsLargetimeSABR}
	Let $(X,Y)$ denote the uncorrelated SABR model (i.e. we set $\rho=0$ in \eqref{eq:SABRSDE}) 
and let $(\overline{X},Y)$ denote the SABR-Brownian motion \eqref{eq:SABRSDEwithDrift} with $\beta\in[0,1)$ and $\rho\in(-1,1)$.
Then in both cases the limit
\begin{equation}\begin{array}{lll}
\lim_{t\to\infty}(X_t,Y_t)=:(X_\infty,Y_\infty)&\textrm{resp.} &\lim_{t\to\infty}(\overline{X}_t,Y_t)=:(\overline{X}_\infty,Y_\infty)
\end{array}\end{equation}
exists almost surely and it holds that $0<\PP(X_\infty>0)<1$ resp. $0<\PP(\overline{X}_\infty>0)<1$.
\end{theorem}
The time change reduces the characterization of the limiting behaviour of the process \eqref{eq:SABRSDE} with $\beta=0, \rho\in(-1,1)$ on a hyperbolic plane 
to determining the hitting time of the coordinate axis of two correlated Brownian motions in the first quadrant of the (Euclidean) plane. Therefore, a lower bound for the probability $\PP(\overline{X}_\infty>0)$ for the SABR-Brownian motion
\eqref{eq:SABRSDEwithDrift} follows from \cite{BranchingModel}.\\

From a financial perspective, such time change constructions can be used to investigate whether the price process has the potential to hit zero in finite time. Such properties have implications for option prices in the limits of extreme strikes as already remarked in \cite{Hobson}. Indeed, the Moment Formula of~\cite{RogerLee} relates the behavior of the implied volatility for extreme strikes to the price of a Put (resp. Call) option. This model-independent result was refined in \cite{BenaimFriz, GulisashviliAsympt10} and extended in~\cite{DMHJ, GulisashviliMass} to the case where the price process has positive mass at zero.
Once the probability mass at the origin is known, arbitrage free wing asymptotics can be derived. Naturally, the probability mass at zero in the SABR model is dependent on the chosen parameter configuration. Asymptotic formulae for the mass at zero in the uncorrelated $\rho=0$ and for the normal $\beta=0$ SABR models were calculated in \cite{SABRMassZero}. 
 
\subsection{Short-time asymptotics and a generalized distance}\label{Sec:AsymptoticsShorttime}
At the heart of heat-kernel type asymptotic expansions lies Varadhan's  classical formula \cite{Varadhan}, which characterizes (for non-degenerate i.e. uniformly elliptic diffusions) the short time asymptotic behaviour of the transition density $p$ at leading order as
\begin{align}\label{eq:VaradhanFormula}
\lim_{t \rightarrow 0} t \log p_t(x,y)=-\frac{d(x,y)^2}{2},\qquad x,y \in \Ss.
\end{align}
The distance $d(\cdot,\cdot)$ appearing on the right hand side of \eqref{eq:VaradhanFormula} is the Riemannian distance
\begin{equation}\label{eq:RiemannianDistance}
d(x,y)=\textrm{inf}\left\{\int_0^1
\sqrt{\langle \dot{\gamma_i}(t),\dot{\gamma_j}(t)\rangle}_{g^{i,j}(\gamma(t))}
dt: \ \gamma\in C^1([0,1]),\ \gamma(0)=x, \gamma(1)=y \right\} 
\end{equation}
induced by the Riemannian metric $g(\cdot,\cdot)$, whose respective coefficients are obtained from the inverse of the covariance matrix of the diffusion (See section \ref{Sec:TimeChangeLocalTime} below for more details). 
The term in the integral is the length of the gradient  vector on a (minimal) parametrised curve\footnote{Note that the length of the parametrised curve (and hence the Riemannian distance) is invariant under reparametrisation and it is conventional to parametrise to arc length (cf. \cite[Section 6]{lee1997riemannian}). In this case, the gradient equals one, and hence the Eikonal equation (which is the starting point of the analysis in \cite{BBF}) is satisfied along the whole curve.}
 from $x$ to $y$.
There are extensions of this result to more general diffusions: see \cite{Molchanov} and L\'{e}andre's extension to the hypoelliptic diffusions \cite{Leandre1,Leandre2,Leandre3}).
\begin{remark}An important observation (and warning) here is that neither of the models discussed in the previous sections (the SABR model, the SABR-Brownian motion or the CEV model) are uniformly elliptic (nor hypoelliptic) in a neighbourhood  of $\{(x,y): x=0, y>0 \}$, that is whenever the forward rate $x$ is near zero.
This lack of regularity is crucial, because the derivation of the SABR formula presented in \cite{HLW}, relies heavily on Varadhan's formula \eqref{eq:VaradhanFormula} and the related heat-kernel expansion. The SABR formula is known to break down around zero forward, where the regularity, necessary for \eqref{eq:VaradhanFormula}, fails.
In fact, to date no \emph{direct} extension of the formula \eqref{eq:VaradhanFormula} is known to be valid for SABR-type models in the neighbourhood of zero.
\end{remark}
For diffusions with degeneracies beyond the hypoelliptic setup, similar asymptotic results can be made in some cases. A more general degenerate setup (general enough to cover the SABR and CEV processes) often requires a suitable generalization of the intrinsic metric \eqref{eq:RiemannianDistance}, via Dirichlet forms. 
Short-time asymptotic results for general degenerate processes are discussed in a number of works \cite{TerElstDiffusion, HinoRamirez, SturmGeometricAspect}. 
Specifically, results of \cite{TerElstDiffusion} indicate that for large classes of degenerate elliptic diffusions the asymptotic relation \eqref{eq:VaradhanFormula} fails, and difficulties often arise from non-ergodic behavior.
A series of articles (e.g. \cite{terElstDegenerateElliptic,terElstSeparation,TerElstDiffusion}) studies the behaviour of second-order operators of the form
\begin{equation*}\label{DegOperator}\begin{array}{ll}
\cA=-\sum_{i,j=1}^{d}\partial_i (\xi_{i,j}) \partial_j
\end{array}\end{equation*}
on $\mathbb{R}^d$ with bounded real symmetric measurable coefficients, such that
$\xi_{i,j}\geq0$ almost everywhere. The remarkable novelty in these works lies in the latter (particularly weak) requirement on the coefficients $(\xi_{i,j})_{i,j}$. Therefore, they cover a large class of degenerate models beyond the hypoelliptic setup.
The setup includes a univariate operator
\begin{equation}\label{eq:terElstOperator}
\partial_x \left(\frac{x^{2\beta}}{(1+x^2)^{\beta}}\right)\partial_x, 
\end{equation}
for $\beta\geq 0$, which approximates for $x\sim0$ the generator 
of the CEV model.
It is shown in \cite{TerElstDiffusion}, that for $\beta\in[0,1/2)$ the validity of the asymptotic relation \eqref{eq:VaradhanFormula} prevails for \eqref{eq:terElstOperator} and for $\beta\geq 1/2$ it fails. In the latter case, the origin is naturally absorbing, and hence acts as an \emph{inpenetrable boundary} although the Riemannian distance $d(-x,x)$ is finite\footnote{Since $d(-x,x)\leq \lim_{a\rightarrow 0}d(-x,a)+\lim_{a\rightarrow 0}d(a,x)$, where both limits are finite.} for any $x>0$, see \cite{TerElstDiffusion}. Indeed, it is well known that a similar phenomenon holds for the CEV model: The origin is naturally absorbing for $\beta\geq1/2$ see \cite{Antonov}, although the Riemannian distance to the origin is finite, see Remark \ref{Rem:CEV}. 
In \cite{TerElstDiffusion} a more general metric is proposed, which better reflects the behavior of the diffusion at zero. In the Appendix \ref{Sec:DiffusionasDirichletform} we included a precise formulation of the intrinsic metric considered in \cite{TerElstDiffusion} for the 
operator \eqref{eq:terElstOperator} for easy reference.
The approach taken in \cite{TerElstDiffusion} was introduced in \cite{HinoRamirez}, where a set-theoretic version of 
Varadhan's formula is proven for open sets $A,B \subset \mathbb{R}^2$
\begin{align}\label{VaradhanFormulaIntegratedHR}
\begin{split}
\lim_{t \rightarrow 0} t \log & \ p_t(A,B)=-\frac{d(A,B)^2}{2}.
\end{split}
\end{align}
Short-time asymptotic results of this sort as in  \cite{TerElstDiffusion, HinoRamirez}
differ from Varadhan's classical short-time asymptotic formula \eqref{eq:VaradhanFormula} in the crucial fact, that \eqref{VaradhanFormulaIntegratedHR} a priori only holds for open sets $A,B \in \R^2$. Therefore the classical formula \eqref{eq:VaradhanFormula} cannot be deduced from such results, unless further regularity conditions are satisfied.
For processes with CEV-type (or SABR-type) degeneracies, no general formula is known which guarantees the validity of a pointwise formula \eqref{eq:VaradhanFormula}. Its weaker version \eqref{VaradhanFormulaIntegratedHR} however, is known to extend to a class of such diffusions which exhibit the same type of degeneracy at the origin as the CEV model, see \cite{TerElstDiffusion}.
\\
In \eqref{VaradhanFormulaIntegratedHR}, the set theoretic distance $d(A,B)$ is, as usual, the infimum of $d(x,y)$, $x\in A, y\in B$, for a suitable intrinsic metric $d(x,y)$ on $\R^d$, determined by the Dirichlet form. Under certain regularity assumptions on the latter (see \cite[Section 1]{SturmLocal} for details) the intrinsic metric can be written in the following form:
\begin{equation*}
d(x,y)=\textrm{sup}\left\{u(x)-u(y):u\in \mathcal{D}_{loc}(\cE)\cap C(X), \frac{d\Gamma(u,u)}{dm}\leq 1\right\}, 
\end{equation*}
where the gradient is generalized to the so-called energy measure (see Appendix \ref{Sec:Closability} for the definition and the monographs \cite{BouleauHirsch,FukushimaOshimaTakeda} for a comprehensive discussion thereof). That is, the energy measure $\Gamma$ is absolutely continuous with respect to the speed measure $m$, and the density 
\begin{align}\label{eq:GenEikonal}
\frac{d\Gamma(u,u)}{dm}(z)
\end{align}
is interpreted as the square of the length of the gradient of $u$ at $z\in X$, cf. \cite{SturmLocal}.
\begin{corollary}[to Theorems \ref{Th:SABRDirichletForm} and \ref{Th:SABRDirichletFormTimechanged} and Lemma \ref{Th:CEVSymmDirichletform}]
The energy measures for the SABR and CEV Dirichlet forms  \eqref{eq:DirichletformSABRwithdrift1} and \eqref{eq:DirichletformSABRwithdrift3} such as the time-changed Dirichlet form
\eqref{eq:DirichletformSABRwithdrift2} are the operators \eqref{eq:EnergymeasureCEV}, and\eqref{eq:Energymeasure2} respectively. 
On $\Ss$ (resp. $(0,\infty)$) these energy measures are in fact determined by the gradient of a geodesic curve on the respective (weighted) manifolds where the Riemannian metric is as in \eqref{eq:ManifoldMetric} for SABR and as in Remark \ref{Rem:CEV} for CEV.
\end{corollary}
See Theorem \ref{closability} in Appendix \ref{Sec:Closability} for a proof of the statements about the energy measure and \cite[pages 108-109: ``Second Proof'']{grigoryan09AnalysisManifolds} to compare with the gradient on a weighted manifold.
Note that in the time-changed forms the energy measures and hence the intrinsic metric (Dirichlet distances) change accordingly as predicted.
\begin{remark}
Results in the spirit of \cite{TerElstDiffusion} leading to Varadhan-type asymptotics often rely on the assumption that the topology induced by the intrinsic metric is equivalent to the original topology of the underlying space $X$ and that all balls  $B_r(x)$, of radius $r>0$, $x,\in X$ are relatively compact, cf. \cite[Section 2, Condition L]{TerElstDiffusion} \footnote{Note that if the topologies coincide (which is the case for the Riemannian distance, cf. \cite[Proposition 11.20]{lee2003introduction}), the fact that all balls are relatively compact is equivalent to the completeness of the metric space cf. \cite{SturmLocal}, which implies in the case of Riemannian manifolds (by the Hopf-Rinow theorem, cf. \cite[Theorem 6.13]{lee1997riemannian}) that the manifold has no boundary.}. 
This is the case for the manifold $(\Ss_1,g_1)$ considered in \eqref{eq:ManifoldMetricBeta10}, but not the case for 
$(\Ss,g)$ in \eqref{eq:ManifoldMetric} for $\beta\in[0,1)$, since the Riemannian distance to the coordinate axis $\{(x,y)\in \R^2: x=0\}$ is infinite for $\beta=1$, but finite for $\beta\in[0,1)$.
\end{remark}

\appendix
\section{Proofs}\label{Sec:Proofs}
\begin{proposition}\label{Th:WellPosednessMPSABR}
Imposing absorbing boundary conditions at $X=0$, the SABR martingale problem is well-posed for any $\beta\in[0,1]$. That is, the process
\begin{equation}\label{eq:SABRMartingaleProblem}
	M_t^f:=f(X_t,Y_t)-f(X_0,Y_0)-\int_0^t A f(X_s,Y_s)\,ds,\qquad f\in C_c^{\infty}(D),
\end{equation}
with $A$ as in \eqref{eq:GeneratorTimeChangedSABR} 
is a martingale for any smooth function with compact support, and uniqueness in law holds for solutions of \eqref{eq:SABRSDE} for any initial value $(x,y) \in D$.
Furthermore, the solutions $P^{(x,y)}$, $x,y\geq0$ form a (strong) Markov family, and even pathwise uniqueness holds.
\end{proposition}
\begin{proof}[Proof of Proposition \ref{Th:WellPosednessMPSABR}]
From It\^o's formula it follows for SABR that \eqref{eq:SABRMartingaleProblem} is a local martingale\footnote{In fact, by \cite[ Chapter 4.8. p.228]{EthierKurtz}, this is sufficient for the subsequent time change arguments.}. 
The statements about pathwise uniqueness follows by local Lipschitz continuity of the coefficients in \eqref{eq:SABRSDE} away from the origin. Well-posedness of the martingale problem and uniqueness is then a consequence of the Yamada-Watanabe theorem, see for instance \cite[Lemma 21.17]{kallenberg2002foundations}.
Finally the strong Markov property follows from the well-posedness of the martingale problem \cite[Theorem 21.11, p.421.]{kallenberg2002foundations}.
\end{proof}
\begin{remark}
In fact, for $\beta\in [\frac{1}{2},1]$, the SABR-martingale problem is well-posed without imposing any boundary conditions at zero, which essentially follows from the well-posedness of the CEV-martingale problem for these parameters, see \cite{AntonovFree}.
For the parameters $\beta\in [0,\frac{1}{2})$, uniqueness holds by local-Lipschitz continuity of the parameters on $\R_{>0}\times\R_{>0}$, and imposing absorbing boundary conditions at $X=0$, uniqueness holds on the whole state space $\R_{\geq 0}\times\R_{>0}$.
\end{remark}
\begin{proof}[Proof of Theorems \ref{Th:TimechangedSABRVolkonskii} and \ref{Th:TimechangedSABRwithDrift}]
The statement of Theorems \ref{Th:TimechangedSABRVolkonskii} and \ref{Th:TimechangedSABRwithDrift} follows from the Dambis-Dubins-Schwarz Theorem \cite[Theorem 1]{KallsenShiryaev} or \cite[Proposition 21.13]{kallenberg2002foundations} directly, and it is immediate that the time change $t\mapsto \int_0^t(Y_s)^2ds$ is nothing but a time change of a Brownian motion $\widetilde Y$ started at $y>0$ to a geometric Brownian motion $Y$.
\end{proof}
\begin{remark}\label{Rem:TimechangeCEV}
We remark here briefly that the corresponding time change, which transforms the a Brownian motion $(W_t)_{t\geq0}$ started at $x>0$ into a CEV process with parameter $\beta\in [0,1]$ reads
$t\mapsto \int_0^t|W_s|^{2\beta}ds$ with corresponding stopping time 
\begin{align*}
\tau_t^{\beta}=\inf\left\{u:\left(\int_0^{\cdot}|W_s|^{2\beta}ds\right)_u\geq t\right\}.
\end{align*}
Note also that a similar scaling can be induced following \cite[Theorem 2]{KallsenShiryaev}, replacing the Brownian motion $Z$ by a symmetric $\alpha$-stable process, see \cite[Definition 1]{KallsenShiryaev}. The corresponding time change is $t\mapsto \int_0^t|Y_s|^{\alpha}ds$
and for $\alpha=2$ the symmetric $\alpha$-stable L\'{e}vy motion in \cite[Theorem 2]{KallsenShiryaev} is a Brownian motion. 
\end{remark}

\begin{proof}[Proof of Lemma \ref{Th:FellerPropertySABR}]
The martingale problem is well-posed so the SABR process is (strong) Markov process and the corresponding semigroup $P_t f(x,y)=E^{x,y}[f(X_t,Y_t)], t\geq 0,$ automatically satisfies a version of the Feller property with weak continuity properties (see for instance Walsh \cite{Wa}). It is a consequence of the well-posedness and the general theory of martingale problems that $P_t$ satisfies the (simple) Feller property: 
the semigroup properties (F1) are implied by the Chapman-Kolmogorov equations of the Markov process, the pointwise continuity (F2) is a consequence of the continuity of paths and 
the property (F), that for any fixed $t\geq0$ $P_t$ maps $C_b(D)$ back into $C_b(D)$ is equivalent (cf. \cite[Lemma 19.3]{kallenberg2002foundations}) to convergence in distribution $(X_t^{x'},Y_t^{y'})\rightarrow (X_t^{x},Y_t^{y})$ as $(x',y')\rightarrow(x,y)$, which is implied by uniqueness in law for solutions of \eqref{eq:SABRSDE}, cf. \cite[Theorem 21.9]{kallenberg2002foundations}.
\end{proof}
Note that if the coefficients of the SDE were bounded, then well-posedness of the martingale problem would readily yield the Feller-Dynkin property cf. \cite[Theorem 21.11.]{kallenberg2002foundations}. Since the coefficitents of the SDE \eqref{eq:SABRSDE} are unbounded, the Feller-Dynkin property needs to be proven separately.

\subsubsection*{Proof of the Feller-Dynkin property for SABR and for SABR-Brownian motion} 
In this section we use the SABR time change in order to prove the Feller-Dynkin property of SABR. By the time change one relates SABR to a CEV model running on a stochastic clock, which in turn allows us to derive the Feller-Dynkin property for the SABR model from Feller's boundary classification of the boundary at infinity for the CEV process.
The property, which distinguishes a Feller-Dynkin process (FD) from a simple Feller process (F) is the 
requirement that 
for any $f\in C_{\infty}[0,\infty)$ and any $t\geq 0$ the following convergence property is satisfied:
\begin{equation}\label{CEVFellerDynkin}
 \lim_{x\rightarrow\infty}\mathbb{E}^x\left[f\left(\widetilde X_{t}\right)\right]=0.
\end{equation}
Hence, the Feller-Dynkin property for SABR is established by the following proposition:
\begin{proposition}\label{Prop2}
Let $\widetilde X$ be the CEV process \eqref{eq:CEVProcess} with parameters $\beta \in[0,1]$ and $\sigma>0$ on a stochastic basis $(\Omega,\mathcal{F},(\mathcal{F})_{t\geq0},\mathbb{P})$ 
resp. the process\footnote{The process $\overline{X}$ corresponds to a Stratonovich-version of the CEV process $\widetilde X$, cf. \cite[Chapter 5, Example 2, p.284]{Protter}} $\overline{X}$ corresponding to the first coordinate of SABR-Brownian motion \eqref{eq:SABRSDEwithDrift}
\begin{equation}\label{eq:CEVProcess}
\begin{array}{rrlrl}
&d\widetilde X_t&=\sigma \widetilde X_t^{\beta}dW_t, & t> 0,&\widetilde X_0=\widetilde x>0\\
\textrm{resp.}&d\overline{X}_t&=\sigma \overline{X}_t^{\beta}dW_t+ \frac{\sigma^2}{2}  \beta\overline{X}_t^{2\beta-1}dt, & t> 0,
&\overline{X}_0=\overline{x}>0
\end{array}
\end{equation}
Furthermore, for any $t\geq 0$ and $y\geq0$, let  $(\tau_t^y)_{t\geq0, y\in[0,\infty)}$ be a family of $\mathbb{P}$-a.s. finite $(\mathcal{F})_{t\geq0}$-stopping times, 
such that for any $t\geq 0$
\begin{equation}\label{ConvergenceOfTimechanges}
\tau_t^{\widetilde{y}} \stackrel{P}{\longrightarrow} \tau_t^y, \quad \textrm{as} \quad \widetilde{y} \rightarrow y.
\end{equation}
Then for any $f\in C_{\infty}[0,\infty)$ and $y\in [0,\infty)$
the following convergence statements hold:
\begin{equation}\label{StatementProp2}
 \lim_{x\rightarrow\infty,\ \widetilde{y}\rightarrow y}\mathbb{E}^x\left[f\left(\widetilde X_{\tau_t^{\widetilde{y}}}\right)\right]=0\qquad \textrm{resp.} \qquad  \lim_{x\rightarrow\infty,\ \widetilde{y}\rightarrow y}\mathbb{E}^x\left[f\left(\overline X_{\tau_t^{\widetilde{y}}}\right)\right]=0.
\end{equation}
\end{proposition}
\begin{proof}
For notational simplicity we only consider $\widetilde X$ here. The statement  \eqref{StatementProp2} follows if for any $\epsilon>0$ and $r>$ there exist constants $N:=N(\epsilon,r)>0$ and $\widetilde{\delta}:=\widetilde{\delta}(\epsilon,r)$, such that
\begin{equation}\label{SIC}
\mathbb{P}^x\left[\widetilde X_{\tau_t^{\widetilde{y}}}\leq r\right] <\epsilon \quad \textrm{for all}\quad x\geq N,\widetilde{y}\in B_y(\widetilde \delta).
\end{equation}
Let us first consider $\tau_t^{\widetilde{y}}$ deterministic, say $\tau_t^{\widetilde{y}}=\overline{t}<\infty$. Then indeed, for any $f \in C_{\infty}[0,\infty)$ and any $\epsilon>0$, let  $\overline{f}:=\sup_{s \in [0,\infty)}f(s)$ and $[0,r]\subset [0,\infty)$ denote the compact set such that 
$|f(s)|<\epsilon$ for $s\notin [0,r]$. Then
\begin{equation*}\mathbb{E}^x\left[f\left(\widetilde X_{\overline{t}}\right)\right]=\mathbb{E}^x\left[f\left(\widetilde X_{\overline{t}}\right)1_{\{\widetilde X_{\overline{t}}>r\}}\right]+\mathbb{E}^x\left[f\left(\widetilde X_{\overline{t}}\right)1_{\{ \widetilde X_t\leq r\}}\right]
\leq \epsilon+\overline{f} \ \mathbb{P}^x\left[\widetilde X_{\overline{t}}\leq r\right].
\end{equation*}
Now let $\epsilon>0$ be arbitrary but fixed. Since $\tau_t^{\widetilde{y}},\tau_t^{y}$ a.s. finite,
there exists a $T>0$ such that 
\begin{equation}\label{eq:randomtimefinite}
\mathbb{P}\left[\tau_t^y >T\right]<\frac{\epsilon}{3},
\end{equation}
and since $\tau_t^{\widetilde{y}},\tau_t^{y}$ are stopping times, for any $\delta>0$ the sets $\{|\tau_t^{\widetilde{y}}-\tau_t^{y}|>\delta\}$ 
are measurable. 
By (\ref{ConvergenceOfTimechanges}) there exists a $\delta>0$ such that 
\begin{equation}\label{DifferenceTauDelta}
\mathbb{P}\left[|\tau_t^{\widetilde{y}}-\tau_t^{y}|>\delta\right]<\frac{\epsilon}{3}. 
\end{equation}
Then for any $r>0$ there exists an $N(\epsilon,r, T)>0$, such that 
\begin{equation}\label{eq:CalcSIC}\begin{array}{rl}
\mathbb{P}^x\left[\widetilde X_{\tau_t^{\widetilde{y}}}\leq r\right]
= &\mathbb{P}^x\left[ \{\widetilde X_{\tau_t^{\widetilde{y}}}\leq r\}\cap \{\tau_t^{y}\leq T\} \right]
+\mathbb{P}^x\left[ \{\widetilde X_{\tau_t^{\widetilde{y}}}\leq r\}\cap \{\tau_t^{y}> T\} \right]
\\\stackrel{\eqref{eq:randomtimefinite}}{<}
&\mathbb{P}^x\left[\{\widetilde X_{\tau_t^{\widetilde{y}}}\leq r\}
\cap\{\tau_t^{y}\leq T\}
\cap\{ |\tau_t^{\widetilde{y}}-\tau_t^{y}|\leq\delta \}\right]\\
& +\mathbb{P}^x\left[\{\widetilde X_{\tau_t^{\widetilde{y}}}\leq r\}
\cap\{\tau_t^{y}\leq T\}
\cap \{ |\tau_t^{\widetilde{y}}-\tau_t^{y}|>\delta \}\right]
+\frac{\epsilon}{3}\\
\stackrel{(\ref{DifferenceTauDelta})}{<}
&\mathbb{P}^x\left[\{\widetilde X_{\tau_t^{\widetilde{y}}}\leq r\}
\cap\{\tau_t^{y}\leq T\}
\cap\{ |\tau_t^{\widetilde{y}}-\tau_t^{y}|\leq \delta \}\right]+\frac{\epsilon}{3}+\frac{\epsilon}{3}\\
\leq\
&\mathbb{P}^x\left[\{\widetilde X_{\tau_t^{\widetilde{y}}}\leq r\}
\cap\{ \tau_t^{\widetilde{y}}\leq T+\delta \}\right]+\frac{\epsilon}{3}+\frac{\epsilon}{3}\\
\leq\
&\mathbb{P}^x\left[ \exists \ \widetilde t \in [0,T+\delta]: \widetilde X_{\widetilde t} \in [0,r] \right]+\frac{\epsilon}{3}+\frac{\epsilon}{3}.
\end{array}\end{equation}
The last probability in \eqref{eq:CalcSIC} coincides with the probability of hitting $[0,r]$ before $T+\delta$
\begin{equation*}
\mathbb{P}^x\left[ \exists \ \widetilde t \in [0,T+\delta]: \widetilde X_{\widetilde t} \in [0,r] \right]=
\mathbb{P}^x\left[T_{[0,r]}^{\widetilde X}\leq T+\delta \right]. 
\end{equation*}
Hence, \eqref{SIC} follows if for any $\epsilon,r, \widetilde T:=T+\delta>0$ there exists a constant $N(\epsilon,r, \widetilde T)>0$, such that 
\begin{equation}\label{eq:ReducedNecessaryCond}
\mathbb{P}^x\left[T_{[0,r]}^{\widetilde X}\leq \widetilde T \right]<\frac{\epsilon}{3} 
\end{equation}
for all $x\geq N(\epsilon,r,\widetilde T)$.
For regular diffusions on an interval $[a,b]\subset \R$, Feller's boundary classification  \cite[ p. 461, eq. (21)]{kallenberg2002foundations}
provides a sufficient condition for \eqref{eq:ReducedNecessaryCond}:
The value of $\mathbb{P}^x[T^{X}_{[0,r]}\leq \widetilde T]$ is clearly non-decreasing in $\widetilde T$ and by the strong Markov property non-increasing in $x$. Hence if the right boundary point ($b=\infty$) is not of entrance type\footnote{
That is, by \cite[ p. 461]{kallenberg2002foundations} if
 \begin{equation*}\label{EntranceBoundary}
 \lim_{T\rightarrow \infty} \left( \inf_{x>r}\mathbb{P}^x\left[T^{X}_{[0,r]}\leq \widetilde T\right] \right)=0, \quad r>0.
 \end{equation*}}, then it holds in particular for any finite $\widetilde T>0$ that
\begin{align}\label{NotEntranceBoundary}
\lim_{x\rightarrow \infty} \mathbb{P}^x\left[T^{X}_{[0,r]}\leq \widetilde T\right] = \inf_{x>r}\mathbb{P}^x\left[T^{X}_{[0,r]}\leq \widetilde T\right] = 0,
\end{align}
where the first equality holds by monotonicity in $x$, the second by monotonicity in $\widetilde T$.
\end{proof}
\begin{lemma}\label{CEVNotEntrance}
Let $(\Omega,\mathcal{F},(\mathcal{F})_{t\geq0},\mathbb{P})$ be a stochastic basis and let $\widetilde X$ (resp. $\overline{X}$) denote CEV process (resp. the process in \eqref{eq:CEVProcess}) on the interval $[a,b]=[0,\infty]$ with parameters 
$\beta \in[0,1]$ and $\sigma>0$.
Then for any $\beta \in [0,1]$ the right endpoint $b=\infty$ is
not an entrance boundary for $\widetilde X$ (resp. $\overline{X}$). 
\end{lemma}
\begin{proof}
For the CEV process it is well known that the right endpoint $b=\infty$ 
is not an entrance boundary for any $\beta\in[0,1]$. We recall the argument from \cite[Theorem 23.13 (iii)]{kallenberg2002foundations} here. Note that the speed measure $\nu$ of the process $\widetilde X$ has density 
$\sigma(x)^{-2}=x^{-2\beta}$, $x\geq 0$ \cite[p. 458]{kallenberg2002foundations}). Then the claim follows from Feller's boundary classification \cite[Theorem 23.12]{kallenberg2002foundations}) and from
\begin{equation*}
\int_1^\infty x\nu(dx)=\int_1^\infty x^{1-2\beta} dx=
 \begin{cases}
  \infty\quad &\textrm{if} \quad \beta\leq 1\\
 \frac{1}{2\beta-2}\quad  &\textrm{if} \quad \beta> 1.
 \end{cases}
\end{equation*}
For the process $\overline{X}$ in \eqref{eq:CEVProcess} the function $p(x)=\frac{x^{\beta+1}}{\beta+1}$ is scale function  (see \cite[Theorem 23.7, p.456]{kallenberg2002foundations}).
In fact, the process $S_t:=p(\widetilde X_t)$, $t\geq0$ satisfies
\begin{equation}
\begin{array}{rrlrl}
dS_t&=c S_t^{2\beta/(\beta+1)}dW_t, & t>0,&S_0&=p(x),
\end{array}\end{equation}
for $c:=\frac{1}{2}(\beta+1)^{2\beta/(\beta+1)}$ and the arguments from \cite[Theorem 23.13 (iii)]{kallenberg2002foundations} apply to $S$. Furthermore, $2\beta/(\beta+1)\in [0,1]$ whenever $\beta\in[0,1]$ and $b=\infty$ is not entrance.
\end{proof}
\subsubsection*{Proofs of the Generalized Feller property for SABR and SABR-Brownian motion}
According to the following Lemma \ref{Th:ReductionSubEigenspaces}, a key ingredient 
for the characterization of the Banach spaces $\cBp$ in Theorem \ref{Th:GeneralizedFellerPropSABRwithDrift} 
is to construct admissible weight functions $\psi$, which are sub-eigenfunctions (that is, functions which satisfy \eqref{eq:SubEigenspace}) of the respective infinitesimal generators $\Delta_g$ of \eqref{eq:SABRSDEwithDrift} resp. $A$ of \eqref{eq:SABRSDE}.
\begin{lemma}[Reduction to Sub-Eigenspaces]\label{Th:ReductionSubEigenspaces}
Let $A$ denote the infinitesimal generator of the semigroup $(P_t)_{t\geq 0}$ and $\psi$ an admissible weight function. Then property $(\widetilde F 4)$ follows if there exists a constant $\lambda>0$ such that
\begin{equation}\label{eq:SubEigenspace}
A \psi(x) \leq \lambda \psi(x) \qquad \textrm{for all} \quad x \in D.
\end{equation}
\end{lemma}
\begin{proof}
Assume that \eqref{eq:SubEigenspace} holds. Then there is an $\epsilon >0$ such that
\begin{equation}\label{eq:SubEigenspaceGronwall1}
P_t\psi(x)\leq \psi(x)+ \lambda \int_0^tP_s\psi(x)ds, \qquad \textrm{for any}\quad t \in  [0,\epsilon].
\end{equation}
Then Gronwall's inequality yields that
\begin{equation}\label{eq:SubEigenspaceGronwall2}
|P_t\psi(x)|\leq \widetilde \lambda\psi(x),\qquad \textrm{for any} \quad t \in [0,\epsilon], \ x\in D. 
\end{equation}
The definition of the norm
$||\cdot||_{\psi}$ 
yields with  $\lambda':=||f||_{\psi}$ the inequality
\begin{equation*}
f(x)\leq \lambda' \psi(x) \qquad \textrm{for all} \quad x\in D.
\end{equation*} 
Hence, by positivity $(\widetilde F 3)$ of the semigroup and by linearity the following estimate holds:
\begin{equation}
P_t f(x) \leq \lambda' P_t \psi(x) \leq \lambda \psi (x),\qquad \textrm{for any} \quad x \in D, \ t \in [0,\epsilon]
\end{equation}
holds, where we used \eqref{eq:SubEigenspaceGronwall2} in the second inequality seting $\lambda:=\lambda' \widetilde \lambda$.
\end{proof}
\begin{lemma}[An admissible weight function $\psi$: The ad-hoc approach]
The function
\begin{align*}
 &\psi: \quad \Ss \quad  \longrightarrow \quad (0,\infty),\\
       &\psi(x,y):=y+2x^{1-\beta}+\frac{x^{2-2\beta}}{y} 
\end{align*}
is a sub-eigenfunction of the SABR infinitesimal $A$, that is
 \begin{align*}
y^2 \left( x^{2\beta} \tfrac{\partial^2 }{\partial x^2} + 2 \rho x^{\beta} \tfrac{\partial^2}{\partial x\partial y}
    + \tfrac{\partial^2 }{\partial y^2} \right)\psi(x,y)\leq 2 \psi(x,y) \quad  \textrm{holds for all} \quad x,y \in \Ss.  
 \end{align*}
\end{lemma}
\begin{proof}The derivatives of $\psi$ are 
\begin{align*}
 \partial_{xx} \psi(x,y) & = (2-2\beta)(1-2\beta) \tfrac{x^{-2\beta}}{y} -\beta(1-\beta)2 x^{-1-\beta}\\
 \partial_{xy}^2 \psi(x,y) & = -(2-2\beta) \tfrac{x^{1-2 \beta}}{y^2}\\
 \partial_{yy} \psi(x,y) & = 2 \tfrac{x^{2-2\beta}}{y^3}.
\end{align*}
Therefore,
\begin{align*}
A \psi(x,y)
&= \left(  (2-2\beta)(1-2\beta)y -\beta(1-\beta)2 x^{\beta-1}y^2 -2 \rho (2-2\beta) x^{1-\beta}+ 
 2\tfrac{x^{2-2\beta}}{y}\right) \\
 &\leq 2 y+4x^{1-\beta}+2\frac{x^{2-2\beta}}{y}=2 \psi(x,y),
 \end{align*}
The claimed inequality follows from the estimates
\begin{align*}
&\quad \ (2-2\beta)(1-2\beta) \leq 2  \quad \quad \quad \quad \textrm{for all} \ \beta \in [0,1]\\
&-\beta (1-\beta) 2 x^{\beta-1}y^2\leq 0 \quad \quad \quad \quad \textrm{for all} \ x,y  \in \R_{\geq 0}\times\R_+ \quad \textrm{and all} \ \beta \in [0,1]\\
&-2 \rho (2-2\beta) \ x^{1-\beta} \ \leq 4 \ x^{1-\beta}\quad \ \textrm{for all} \ x  \in \R_{\geq 0}, \ \textrm{all} \ \beta \in [0,1]\ \textrm{and all} \ \rho \in [-1,1].
\end{align*}
\end{proof}
The ad-hoc approach yields a suitable weight function for all SABR parameters. However, this function grows only at a rate $x^{2-2\beta}$, and therefore with this choice of $\psi$ for parameters $\beta>\frac{1}{2}$ the space $\cBp(D)$ does not include the payoff function of a European call option.
In order to extend the above to higher exponents, we take a geometric approach:
We determine true eigenspaces of the hyperbolic Laplace-Beltrami operator and from these we construct suitable eigenspaces of the generator of the SABR-heat semigroup such as sub-eigenspaces of the above type---under suitable parameter restrictions---for arbitrarily high (integer) exponents for the SABR infinitesimal generator. 
Indeed, in \cite{HLW} a local isometry is introduced from the SABR-plane to the Poincar\'{e}-plane
\begin{equation}\label{eq:SABRIsometry}
\begin{array}{lrll}
\phi: & (\mathcal{S},g) & \longrightarrow (\mathbb{H},h), \\
& (\widehat{x},\widehat{y}) & \longmapsto \left(x,y\right)
 :=
\displaystyle \left(\frac{\widehat{x}^{1-\beta}}{\rrho(1-\beta)} - \frac{\rho \widehat{y}}{\rrho}, \widehat{y} \right). 
\end{array}
\end{equation}

\begin{lemma}[Radial eigenspaces on the Poincar\'{e} plane]\label{Th:RadialEigenspacesHyperbolic}
Radial eigenfunctions of hyperbolic Laplace-Beltrami operator $\Delta_h$ are of the form
\begin{align*}
L_{\lambda}(r_Z(z))=L_{\lambda}(\cosh (d(Z,z))=L_{\lambda}\left(1 + \frac{(x - X)^2 + (y - Y)^2}{2yY}\right), \quad z\in \mathbb{H}^2,
\end{align*}
where $Z=(X,Y)\in \mathbb{H}$ is an arbitrary fixed reference point, $d$ denotes the hyperbolic distance
\begin{equation}\label{eq:radialdistance}
d(Z,z)=\textrm{arcosh}\left(1 + \frac{(x - X)^2 + (y - Y)^2}{2yY}\right), \quad Z,z \in \mathbb{H}
\end{equation}
and $L_{\lambda}$ is a Legendre polynomial with $\lambda=-n(n+1)$, for some $n\in \mathbb{N}$. 
\end{lemma}
\begin{proof}
Radial eigenfunctions of the hyperbolic Laplace-Beltrami operator $\Delta_h$ are well known. See \cite[pages 82 and 275]{grigoryan09AnalysisManifolds}. 
\end{proof}
\begin{example}\label{Ex:HyperbolicRadialEigenfunction}
An obvious example is the hyperbolic distance from a reference point $Z=(0,1)$, on a hyperbolic cosine scale, is the function 
\begin{align*}
 \tilde r_0(z):= \cosh d(z,Z) = 1 + \frac{x^2 + (y - 1)^2}{2y}
,\qquad z\in \HH^2.
\end{align*}
This function is an eigenfunction of $\Delta_h$ satisfying $\Delta_h \ \tilde r_0=2\ \tilde r_0$. 
\end{example}
The local isometry property of \eqref{eq:SABRIsometry} and an application of \cite[Lemma 3.27]{grigoryan09AnalysisManifolds} (see \eqref{eq:EigenfunctionPullback} below) allows us to construct eigenfunctions of the Laplace-Beltrami operator $\Delta_g$ of the SABR plane from eigenfunctions 
of the Laplace-Beltrami operator $\Delta_h$ of the hyperbolic plane.
\begin{lemma}[Implied eigenspaces on the SABR-plane]\label{Th:RadialEigenspacesSABR}
For any radial eigenfunction $\psi$ with eigenvalue $\lambda$ of the Laplace-Beltrami operator $\Delta_h$ on the Poincar\'{e}-plane $(\mathbb{H},h)$ the pullback $\phi^*\psi$ under the SABR-isometry \eqref{eq:SABRIsometry} (i.e. the composition $\phi\circ \psi:\Ss\rightarrow \R$) is a radial eigenfunction of the Laplace-Beltrami operator $\Delta_g$ of the SABR-plane $(\Ss,g)$, with $\Ss=\{(\hx,\hy):\hx,\hy>0\}$ to the same eigenvalue.
Therefore, radial eigenfunctions of $\Delta_g$ are of the form 
\begin{equation}\label{eq:SABRLaplEigenfunctionForm}
L_{\lambda}\left( \cosh \delta(\hz,\hZ)) \right)=L_{\lambda}\left(\cosh d(\phi(\hz),\phi(\hZ)) \right),\qquad \hz \in \Ss
\end{equation}
 for some fixed $\hZ=(\hX,\hY) \in \mathcal{S}$, where $L_{\lambda}$ denote Legendre polynomials with $\lambda=-n(n+1)$, for some $n\in \mathbb{N}$, and $\delta$ is the Riemannian distance function on the SABR-plane. 
 Explicitly,
\begin{equation}\label{eq:SABRLaplEigenfunction}
\cosh d(\phi(\hz),\phi(\hZ))= 1+ \frac{ 
\Big( \frac{1}{1-\beta}(\hX^{1-\beta}-\hx^{1-\beta}) - \rho(\hY-\hy)  \Big)^2}
{(1 - \rho^2)2\hY \hy} + \frac{(\hY-\hy)^2}{2\hY \hy}.
\end{equation}
\end{lemma}
\begin{proof}
Let $Z:=\phi(\hZ),z:=\phi(\hz)\in \HH^2$. 
Then from the construction of the distance function 
\begin{equation}
\delta(Z,z)=d(\phi(\hZ),\phi(\hz)), \qquad \hZ,\hz \in \Ss
\end{equation}
on the SABR-plane it is immediate that for any $f:\R_{\geq 0}\rightarrow \R$
\begin{equation}
f(d(Z,z))=f(d(\phi(\hZ),\phi(\hz))=f(\delta(\hZ,\hz)), \qquad \hZ,\hz \in \Ss
\end{equation}
and the form \eqref{eq:SABRLaplEigenfunction} follows from Lemma \ref{Th:RadialEigenspacesHyperbolic} with \eqref{eq:SABRIsometry} directly.
Furthermore, by Lemma \ref{Th:RadialEigenspacesHyperbolic} 
$$\psi_{\lambda}~:=~L_{\lambda}~\circ~\cosh~\circ~d$$ is an eigenfunction of $\Delta_h$ to the eigenvalue $\lambda$, hence by Lemma \ref{Th:RadialEigenspacesHyperbolic} and by \cite[Lemma 3.27]{grigoryan09AnalysisManifolds}, the pullback $\phi^*\psi_{\lambda}$ is an eigenfunction of $\Delta_g$ to the same eigenvalue
\begin{equation}\label{eq:EigenfunctionPullback}
\Delta_g(\phi^*\psi_{\lambda})(\hz)=\Delta_g L_{\lambda} (\delta(\hZ,\hz))\stackrel{
}{=}\Delta_h L_{\lambda} (d(Z,z))=\Delta_h\psi_{\lambda}(z)=\lambda \psi_{\lambda}(z)=\lambda (\phi^*\psi_{\lambda})(\hz).
\end{equation}
\end{proof}
\smallskip
\begin{proof}[Proof of Theorem \ref{Th:GeneralizedFellerPropSABRwithDrift}]
Note that \eqref{eq:Cosh} is nothing but the function \eqref{eq:SABRLaplEigenfunction} with reference point $\phi(\hZ)=(\frac{1}{\rrho^2},1)\in\mathbb{H}$.
Hence the functions \eqref{eq:Cosh} are eigenfunctions of $\Delta_g$, which is an immediate consequence of Lemma \ref{Th:RadialEigenspacesSABR}. 
For $c\in [1,\infty)$ the reference point $\phi(\hZ)=(c/\rrho,1)\in \HH^2$ can always be realized as the image under $\phi$ of a point $\hZ\in \Ss=(0,\infty)\times(0,\infty)$ (that is, there exists $\hZ=(\hX,\hY)\in \Ss$ such that $\phi(\hX,\hY)=(c/\rrho,1)$):
By the definition the map $\phi$ in \eqref{eq:SABRIsometry},
\begin{equation}\label{eq:ReferencePoint}
(\hX,\hY)=\left(\left((1-\beta)(c+\rho) \right)^{1/(1-\beta)},1\right)  \qquad \Rightarrow  \qquad \phi(\hX,\hY)=\left(c/\rrho,1\right),
\end{equation}
and indeed for $c\in [1,\infty)$ the value of $c+\rho$ is always positive. The same statement holds whenever $\rho \in (0,1)$ or when $\rho\in (-1,0)$ but $c>|\rho|$.
If $\rho\in(-1,0)$ and $c\leq |\rho|$, then
it holds that $\hX > 0$, whenever the following conditions on the parameter $\beta$ are satisfied:
$\beta ~\in ~\{0\}\cup\left\{\frac{2m-1}{2m}, m\in \mathbb{N}\right\}.$ In this case, it is ensured that $1/(1-\beta)=2m$ for some $m\in \mathbb{N}$ and hence
$\hX=\left((1-\beta)(c+\rho) \right)^{1/(1-\beta)}>0$ and hence $(\hX,\hY)\in \Ss$ although $c+\rho<0$.
Furthermore, the eigenfunctions $\psi$ are admissible weight functions 
if the range is bounded from below.
The range of the function $r_{\hZ}(\cdot)\equiv\cosh(\delta(\hZ,\cdot))$ is on $(1,\infty)$ cf. \eqref{eq:SABRLaplEigenfunction}, hence composition with the Legendre polynomials $L_{n(n+1)}\left(r_{\hZ}(\cdot)\right)$ yields functions with range on $(1,\infty)$ by the properties of Legendre polynomials, cf. \cite[Section 8]{Abramowitz}.
Having constructed admissible weight functions, which are eigenfunctions of the operator, it follows from Lemma \ref{Th:ReductionSubEigenspaces} that the heat-semigroup \eqref{eq:HeatSemigrouponBpsi} has the generalized Feller property (FG). 
The strong continuity of the SABR-heat semigroup \eqref{eq:HeatSemigrouponBpsi} on the Banach spaces $\mathcal{B}^{\psi_n}$, $n\in \N$ is then a consequence of Theorem \ref{Th:StrongContinuityBpsi}. 
This completes the proof of all statements of Theorem \ref{Th:GeneralizedFellerPropSABRwithDrift} about the SABR-heat semigroup. 
Now for the SABR semigroup, recall that the following relationship holds between $\Delta_g$ and the generator $A$ of the SABR-model
\begin{equation}\label{eq:SABRandLaplaceRel}
 Af=\Delta_g f-\beta y^2 x^{2\beta-1} \frac{\partial f}{\partial x}, \qquad f\in C^{\infty}_c(D).
\end{equation}
Therefore, if an eigenfunction $\psi_{\lambda}$ of the Laplace-Beltrami operator is bounded with respect to the first order
term of \eqref{eq:SABRandLaplaceRel} in the sense \ref{eq:SABRandLaplaceEigenfuncRel} below, then $\psi_{\lambda}$ is also sub-eigenfunction of the SABR infinitesimal generator
\begin{equation}\label{eq:SABRandLaplaceEigenfuncRel}
\Delta_g\psi_{\lambda}=\lambda \psi_{\lambda}\qquad \textrm{and} \qquad A\psi_{\lambda}\leq\widetilde \lambda \psi_{\lambda} \quad \textrm{for some} \quad \lambda, \widetilde \lambda \in \R. 
\end{equation}
Such a condition is
guaranteed to hold if the following \emph{drift condition} is fulfilled: There exists a constant $\widehat \lambda \in \R$ such that
\begin{equation}\label{eq:DriftCondition}
0\leq \beta \hy^2 \hx^{2\beta-1} \frac{\partial \psi_{\lambda}}{\partial \hx} + \widehat{\lambda}\psi_{\lambda} \qquad \textrm{for all}\quad \hx,\hy>0.
\end{equation}
For any $\psi_{\lambda}$, which satisfies \eqref{eq:DriftCondition}, the statement \eqref{eq:SABRandLaplaceEigenfuncRel} holds for any $\widetilde\lambda\geq(\lambda+\widehat \lambda)$, where $\lambda$ is an eigenvalue of $\Delta_g$ and $\widehat{\lambda}\in \R$ is the constant from the drift condition \eqref{eq:DriftCondition}. Hence an eigenfunction of $\Delta_g$ which satisfies \eqref{eq:DriftCondition} is a sub-eigenfunction of $A$
\begin{equation*}
A \psi_{\lambda}=\Delta_g\psi_{\lambda}-\beta \hy^2 \hx^{2\beta-1} \frac{\partial \psi_{\lambda}}{\partial \hx}\leq \widetilde\lambda \psi_{\lambda}.
\end{equation*} 
For such (sub-) eigenfunctions the statement of Theorem \ref{Th:GeneralizedFellerPropSABRwithDrift} about the generalized Feller property (FG) of the SABR semigroup \eqref{eq:Semigroup} follows from Lemma \ref{Th:ReductionSubEigenspaces} analogously as above 
 and strong continuity of the SABR semigroup 
on the Banach spaces $\mathcal{B}^{\psi_n}$, $n\in \N$ is again a consequence of Theorem \ref{Th:StrongContinuityBpsi}.
It remains to find eigenfunctions which satisfy \eqref{eq:DriftCondition} for some $\widehat \lambda \in \R$.
For the functions $\psi_{c,n}:=L_{n(n+1)}\circ r_c$, where $L_{n(n+1)}$ denote Legendre polynomials of order $n(n+1),\ n\in \mathbb{N}$
the drift part in \eqref{eq:SABRandLaplaceRel} is of the form
\begin{align}\label{eq:DriftPositive}\begin{split}
\beta \hy^2 \hx^{2\beta-1} \partial_{\hx}\psi_{c,n}
&=\beta \hy^2 \hx^{2\beta-1}  (\partial_{r}L_{n(n+1)} \ \partial_{\hx} r_c)\\ 
&=\partial_{r}L_{n(n+1)}(r_c)\frac{\beta }{1-\rho^2}\left(\frac{\hy}{(1-\beta)}-\hx^{\beta-1}\hy(\rho\hy+c)\right).
\end{split}\end{align}
The derivatives of the Legendre polynomials\footnote{
Recall that that $L_{0}(r_c)\equiv1$, $L_{2}(r_c)\equiv r_c$ and $L_{6}(r_c)\equiv \frac{1}{2}(3r_c^2-1)$. Generally, it holds for the derivatives that
$(r_c^2-1)\partial_{r}L_{n(n+1)}(r_c)=n(r_cL_{n(n+1)}(r_c)-L_{(n-1)n}(r_c))$, which can be negative for values of $r_c$ in $[0,1)$, but all derivatives are nonnegative on $r_c\in[1,\infty)$.} 
for $n=0,1,2$ are clearly positive for $\partial_{r}L_{0}(r_c)\equiv 0$, $\partial_{r}L_{2}(r_c)\equiv 1$ and $\partial_{r}L_{6}(r_c)\equiv r_c$, which is positive for all $\hx > 0,\ \hy>0$ by construction. In these cases $\psi_{c,n}$ satisfies the drift condition \eqref{eq:DriftCondition} with $\widehat{\lambda}=0$, since if $\rho< 0$ and $c=0$ the expression in \eqref{eq:DriftPositive} is clearly nonnegative for all $\hx > 0,\ \hy>0$ and $\beta\in [0,1)$.
The restriction on $\beta$ stated in $(iii)$ of Theorem \ref{Th:GeneralizedFellerPropSABRwithDrift} is a consequence of the choice $c=0$ and follows from the arguments presented in \ref{eq:ReferencePoint}.
This string of argumentation can be generalized to Legendre polynomials of arbitrary order. Indeed, the derivatives of the Legendre can become negative, but this happens only for $|r_c|\in[0,1)$ and  for $r_c\in[1,\infty)$ they are always positive. On the other hand, the range of \eqref{eq:Cosh} is $r_c(x,y)\in[1,\infty)$ for all $x,y$ and all SABR-parameters $\beta,\rho$.
The latter statement can be read off the equation\footnote{
It can be seen more directly from the representation \eqref{eq:SABRLaplEigenfunction}, where we recall that \eqref{eq:Cosh} is the function \eqref{eq:SABRLaplEigenfunction}, for the choice following choice of $\hZ$: $\phi(\hZ)=(\frac{1}{\rrho^2},1)$.} \eqref{eq:Cosh}, noting $\frac{1+y^2}{2y}\geq1$ for all $y>0$, while the second term in \eqref{eq:Cosh} is always nonnegative.
\end{proof}
\smallskip
\begin{proof}[Proof of Corollary \ref{Th:Lineargrowth}]
Note that whenever $n\in \mathbb{N}$ is large enough such that
$\beta\leq \frac{2n-1}{2n}$, the sub-eigenfunctions 
\begin{equation}\label{eq:SubEigenfunctionForm}
\psi_n(\hx,\hy)=L_{n(n+1)}(r_c(\hx,\hy)), \qquad \textrm{for} \quad \hx,\hy \in \Ss                            
\end{equation}
have at least linear growth in $\hx$. Recall that in \eqref{eq:SubEigenfunctionForm}, the function $r_c$ denotes \eqref{eq:Cosh}, and $L_{n(n+1)}$ is a Legendre polynomial of order $n$ and $c(n)$ is a finite constant depending on $n$, for $n\in \mathbb{N}$.
\end{proof}
\begin{proof}[Proof of Theorem \ref{Th:AsymptoticsLargetimeSABR}]
Recall that for a geometric Brownian motion 
both boundary points $\{0\}$ and $\{\infty\}$ are natural and hence $\mathbb{P}[S_t=0]=0$ for any $t>0$.
Moreover, for the geometric Brownian Motion $Y$ from \eqref{eq:SABRSDE} we have $\lim_{t\rightarrow \infty}Y_t=0$, $\mathbb{P}-a.s.$
In particular, by a slight abuse of notation, $Y_{\infty}=0$, $\mathbb{P}$-a.s., $Y_t>0$ for all $t\in [0,\infty)$, and $\inf\{t>0: Y_t=0\}=\infty$.
Therefore,
\begin{equation}\label{eq:HittingTimeOfZeroY}
0=Y_{\infty}=\wY_{\tau(\infty)} =\wY_{\tau_1}=\wY_{\tau_0}
\end{equation}
where the random times $\tau_0$ and $\tau_1$ are as in \cite[Chapter 6, p. 307]{EthierKurtz}
\begin{equation}\label{eq:Tau0Tau1}
\tau_0:=\inf \left\{u: (\wY_u)^2=0 \right\} \quad \tau_1:=\inf \left\{ u: \int_0^u (\wY_s)^{-2}   ds =\infty \right\}. 
\end{equation}
Note that for almost every sample path $\tau_0=\tau_1$ and $\wY(\tau_0)=0$ when $\tau_0<\infty$, 
since $Y_t>0$ for all $t\in [0,\infty)$, $\tau_1$ is the first time when $\wY$ hits zero, which justifies the last equality in \eqref{eq:HittingTimeOfZeroY}.
In fact, since $\wY$ is a Brownian motion started at $\widetilde y>0$, $\tau_0$, i.e. the first hitting time $T_0^{\wY}$ of zero is finite with positive
probability.
		Since $X$ is a non-negative martingale, $\lim_{t\to\infty} X_t=X_\infty$ exists almost surely. 
		Now decomposing $X$ via time change, the first component becomes a CEV process $\widetilde X$ (on a stochastic clock), and as such it hits zero in finite time
for any $\beta < 1$, 
\cite[Theorem 51.2]{RW2}.
There are three cases, which can occur:
For any $\beta \in (0,1)$, the CEV-provess $\widetilde X$ and the Brownian motion $\widetilde Y$ started at $y$ reach zero a.s.\\ 
\begin{enumerate}
 \item[Case 1] $T_0^{\widetilde X} < T_0^{\wY}$ : In this case, the process $\widetilde X$ hits zero and after hitting zero remains zero until $T_0^{\wY}$, which is the end of our time consideration.
 \item[Case 2] $T_0^{\widetilde X} > T_0^{\wY}$ : In this case, the process $\wY$ hits zero first and ends our time consideration, 
while $\widetilde X$ is still positive (and finite). 
Then the dynamics after $T_0^{\wY}$ become $dX_t=0X_t^{\beta}dW_t=0.$
 \item[Case 3] $T_0^{\widetilde X} = T_0^{\wY}$ : In this case, $\widetilde X$ approaches zero as $\wY$ does and therefore\\
$X$ asymptotically approaches zero for $t \rightarrow \infty$.\\
\end{enumerate}
		To show that $\PP(X_\infty>0)\in (0,1)$ we use the SABR time change: By continuity 
		\begin{align*}
			\lim_{t\to\infty} X_t= \lim_{t\to\infty}\widetilde X_{\tau^{-1}_t}= \widetilde X_{\tau^{-1}_\infty},
		\end{align*}
		Since also the first hitting time $T^{\widetilde X}_0$ of $\widetilde X$ is absolutely continuous with strictly positive density on $[0,\infty)$ it follows that 
\begin{equation}\label{eq:ProbSABRPositiveLimit}
\PP(X_\infty>0)=\PP(\widetilde X_{\tau_\infty^{-1}}>0)=\PP(\tau_\infty^{-1}<T_0^{\widetilde X})\in (0,1). 
\end{equation}
Note that for the proof of the statement about $(\overline{X},Y)$ follows by analogous arguments for \eqref{eq:SABRSDE} and \eqref{eq:SABRdynamicsTimechanged} with $\beta=0$ and general $\rho\in[-1,1]$, since the first coordinate process $\overline{X}$ in \eqref{eq:CEVProcess} has an explicit solution \cite[Example 2, p. 284]{Protter} of power form in $\beta$
\begin{equation}\label{eq:ExplicitsolCEVStrat}
 \overline{X}_t=\left(\overline{x}^{1-\beta}+\sigma^2(1-\beta)W_t\right)^{1/(1-\beta)},
\end{equation}
therefore, the expression \eqref{eq:ExplicitsolCEVStrat} takes the value zero if and only if it vanishes for $\beta=0$. 
The time change (a change of the underlying geometry) reduces the problem of determining the limiting behavior of the process \eqref{eq:SABRSDE} with $\beta=0, \rho\in[-1,1]$, which is a (possibly correlated) Brownian motion on a hyperbolic plane 
to determining the hitting time of the coordinate axis of two correlated Brownian motions in the first quadrant of the (Euclidean) plane, which is available in \cite[equation (5.8)]{BranchingModel}. 
\end{proof}
\begin{proof}[Proof of Lemma \ref{Th:CEVSymmDirichletform}]
The statement follows immediately by integrating by parts \eqref{eq:DirichletformSABRwithdrift3} and applying Theorem \ref{Th:HamzaCondition} to obtain the closability on $C_0^{\infty}((0,\infty))$ and the domain of closedness follows from Lemma \ref{Th:OneDimClosability}.

\end{proof}
\begin{proof}[Proof of Theorems \ref{Th:SABRDirichletFormSymmetricNo}, \ref{Th:SABRDirichletForm} and \ref{Th:SABRDirichletFormTimechanged}]
Let us denote the coefficients of the SABR infinitesimal generator \eqref{eq:GeneratorTimeChangedSABR} by the following matrix
\begin{align}\label{SABRXi}
\begin{split}
\xi:\quad \R^2 &\longrightarrow \quad \R^{2\times2}\\
(x,y)&\longmapsto  \left(\begin{array}{cc}
y^2 x^{2\beta} & y^2 \rho x^{\beta}\\
y^2 \rho x^{\beta} & y^2
\end{array}
\right)
\end{split}
\end{align}
and let us denote the \emph{SABR-bilinear forms}
\begin{align*}
\mathcal{E}_m(u,v)=\frac12 \int_{S}\sum_{i,j}
\xi_{ij}(x) \ \partial_i u(x)\ \partial_j v (x) \ m(x)\ dx 
\end{align*}
appearing in Theorem \ref{Th:SABRDirichletFormSymmetricNo} and in \eqref{eq:DirichletformSABRwithdrift1} of Theorem \ref{Th:SABRDirichletForm} for short.
The essential statements in these theorems are: the closability of the bilinear form (assertion $(i)$ of Theorems \ref{Th:SABRDirichletForm} and \ref{Th:SABRDirichletFormTimechanged}) and the specification of the domain of closedness (assertion $(ii)$ of Theorems \ref{Th:SABRDirichletForm} and \ref{Th:SABRDirichletFormTimechanged}) which can be concluded from \cite[Proposition 1]{BouleauDenis} or \cite[Section 4]{RoecknerWielens} as we will demonstrate here. We included their statement in the Appendix \ref{HamzaGeneral2}) for easy reference.
The Borel function $\xi$ from (\ref{SABRXi}) clearly satisfies Condition (HG2) (see Section \ref{HamzaGeneral2}) for all $\beta\in [0,1]$ and all $\rho\in(-1,1)$. 
Then if there exists a Borel function $m: \R^2 \longrightarrow \R_+$, which satisfies Condition (HG1) (see Section \ref{HamzaGeneral2})
then Proposition \ref{closability} is applicable, i.e. the Bilinear form $\cE_{m}$ is closable on $C_0^{\infty}(S)$ (as a consequence of \cite[Section 4]{RoecknerWielens}) and it follows from \cite[Proposition 1]{BouleauDenis} that the pair
$(\cE_{m}, D(\cE_{m}))$ corresponding with to this choice of 
speed measure $m$ is a Dirichlet form on the 
Hilbert space $\Ll^2(mdx)$.
In particular, if $m$ satisfies Condition (HG1) (see Appendix \ref{HamzaGeneral1}), then $(\mathcal{E}_m,D(\mathcal{E}_m))$ is a closed form, 
and its generator $\Phi(\cE_m)$ is 
a negative self-adjoint operator $A_m$, 
such that 
\begin{align*}
\cE_m(u,v)=(A_mu,v)_{L^2(mdx)} \quad \textrm{for all} \quad v \in D(\cE_m) \quad \textrm{and} \quad u\in D(\Phi(\cE_m)). 
\end{align*}
Indeed by integration by parts we get
\begin{align*}
\mathcal{E}_m(u,v)=&-\frac12\sum_{i,j}\int_{\R^2}\partial_j \Big(\xi_{ij}(x) \partial_i u(x) m(x)\Big) \ v (x)\ dx\\
=&-\frac12 \int_{\R^2}\sum_{i,j} \Bigg( \partial_j\xi_{ij}(x) \ \partial_i u(x) +
\xi_{ij}(x) \partial_j\partial_i u(x)
+\xi_{ij}(x) \partial_i u(x) \frac{\partial_jm(x)}{m(x)}\Bigg) v (x)\ m(x)dx\\
=&(A_mu,v)_{L^2(mdx)},
\end{align*}
from which we can derive the following simple \emph{no drift condition}: The Generator $A_m$ of the Dirichlet form $\cE_m$ has no lower order terms if and only if 
\begin{equation}\label{NoDriftCondition}
\sum_{i}\sum_{j}\Big(\partial_j\xi_{ij}(x) \  +\xi_{ij}(x) \frac{\partial_jm(x)}{m(x)}\Big)\partial_i u(x)=0, 
\end{equation}
for any $u\in D(\Phi(\cE_m))$. 
If the matrix $\xi$ is the matrix of SABR-coefficients \eqref{SABRXi},
then the no drift condition \eqref{NoDriftCondition} implies that the generator $A_m$ of the Dirichlet form $\cE_m$ coincides with 
the SABR infinitesimal generator \eqref{eq:GeneratorTimeChangedSABR} on their common domain. Since $\cE_m$ is closable on $C_0^{\infty}(S)$ in $\Hh_m$, this common domain contains the dense subset $C_0^{\infty}(\Ss)$.
Inserting the values \eqref{SABRXi} into \eqref{NoDriftCondition} yields the following differential equations for $m$:
\begin{align}\label{eq:DiffEqSymmNo}
x^{2\beta-1}y^2\left(2\beta+x\frac{\partial_x m(x,y)}{m(x,y)}\right) + \rho x^{\beta}y\left(2+y\frac{\partial_ym(x,y)}{m(x,y)}\right)&=0\\
\rho x^{\beta-1}y^2\left(\beta+x\frac{\partial_xm(x,y)}{m(x,y)}\right) + y\left(2+y\frac{\partial_ym(x,y)}{m(x,y)}\right)&=0.
\end{align}
Assume that $m$ satisfies \eqref{eq:DiffEqSymmNo}. 
Let us define the auxiliary function
\[
 f(x,y) := \log(m(x,y)).
\]
This function satisfies
\begin{align}
\label{eq:main}\begin{split}
 x^{2\beta-1}y^2\left(2\beta+x\partial_x f\right) + \rho x^{\beta} y \left(2+y\partial_y f\right) = 0,	\\
 \rho x^{\beta-1}y^2\left(\beta+x\partial_x f\right) + y\left(2+y\partial_y f\right) = 0.
\end{split}\end{align}
Multiplying the second equation in \eqref{eq:main} by $\rho x^{\beta}$ and equating to the first yields
\[
 x^{2\beta-1}y^2\left(2\beta+x\partial_x f\right) = \rho^2 x^{2\beta-1}y^2\left(\beta+x\partial_x f\right).
\]
This yields for $x \neq 0 \neq y$
 $x\partial_x f = \frac{(\rho^2-2) \beta}{1-\rho^2}$.
This implies for some function $g$ the identity
\begin{equation}\label{eq:midentity1}
 f(x,y) = a \log(x) + g(y),
\end{equation}
where we denoted $a:= x\partial_x f$.  Similarly, multiplying the first equation in \eqref{eq:main} by $\rho$ and the second by $x^{\beta}$ yields
$\rho x^{2\beta-1}y^2\beta + \rho^2 x^{\beta} y \left(2+y\partial_y f\right) =  yx^{\beta}\left(2+y\partial_y f\right)$.
This yields for $x \neq 0 \neq y$
$
 y\partial_y f = -2 + \frac{\rho \beta x^{\beta-1} y}{1-\rho^2}$.
Solving this we obtain for some function $h$ the identity
\begin{equation}\label{eq:midentity2}
 f(x,y) = -2 \log(y) + \frac{\rho \beta x^{\beta-1} y}{1-\rho^2} + h(x).
\end{equation}
Thus we know that for all $x\neq 0 \neq y$
\begin{equation}
\label{eq:ode}
 a \log(x) + g(y) = -2 \log(y) + \frac{\rho \beta x^{\beta-1} y}{1-\rho^2} + h(x).
\end{equation}
Setting $x=1$ in \eqref{eq:ode} yields for all $y \neq 0$
\begin{equation}
\label{eq:y}
g(y) = -2 \log(y) + \frac{\rho \beta y}{1-\rho^2} + h(1),
\end{equation}
while setting $y=1$ in \eqref{eq:ode} implies that for all $x\neq 0$
\begin{align}
\label{eq:x}
 a \log(x) + g(1) = \frac{\rho \beta x^{\beta-1}}{1-\rho^2} + h(x).
\end{align}
Setting $y=1$ in \eqref{eq:y}  and using \eqref{eq:x} yields
(whenever $x\neq 0$) that
\begin{align}
\label{eq:h}
  h(x) = a \log(x) + \frac{\rho \beta}{1-\rho^2} + h(1) - \frac{\rho \beta x^{\beta-1}}{1-\rho^2}.
\end{align}
Substituting the equations for $g, h$ given by \eqref{eq:y}, \eqref{eq:h} into \eqref{eq:ode} we obtain
\begin{align}\label{eq:criterion}
 \frac{\rho \beta y}{1-\rho^2} + h(1) = \frac{\rho \beta x^{\beta-1} y}{1-\rho^2} + \frac{\rho \beta}{1-\rho^2} + h(1) - \frac{\rho \beta x^{\beta-1}}{1-\rho^2},
\end{align}
for all $\forall \, x\neq 0 \neq y$.
The criterion in equation \eqref{eq:criterion}  
only holds true if $\beta = 1$, $\beta=0$ or $\rho=0$. \\
Now for the last statements about the generator $\Psi(\cE_m)$ of the Dirichlet form consider the following operator:
\begin{equation}\label{eq:LaplBeltramiWeightedManifold}
\Delta_{\Upsilon \mu}=\frac{1}{\Upsilon\mu}\sum_{i}\frac{\partial}{\partial x^i}\left(\sum_{j}\Upsilon\mu \ g^{ij} \frac{\partial}{\partial x^j}\right)
\end{equation}
In the case of SABR coefficients  
one has $\mu(x,y)=\frac{1}{\rrho x^{\beta}y^2}$,
and \eqref{eq:LaplBeltramiWeightedManifold} reads
\begin{align*}
\Delta_{\Upsilon\mu}&=
 \frac{x^{\beta}}{\Upsilon(x,y)}y^2 \left(\partial_x 
\Upsilon(x,y) x^{\beta} \partial_x +\Upsilon(x,y) x^{\beta} \partial_{xx}^2
+\partial_x \Upsilon(x,y) \rho \partial_y + \Upsilon(x,y) \rho \partial_{xy}^2\right)\\
&+ \frac{x^{\beta}y^2}{\Upsilon(x,y)}  \left( \partial_y 
 \Upsilon(x,y)\rho \partial_x +\Upsilon(x,y) \rho \partial_{yx}^2
+\partial_y\frac{\Upsilon(x,y)}{x^{\beta}} \partial_y + \frac{\Upsilon(x,y)}{x^{\beta}} \partial_{yy}^2\right).
\end{align*}
Choosing $\Upsilon\equiv 1$ yields the Laplace-Beltrami operator $\Delta_g$ in \eqref{eq:SABRLaplaceBeltrami}
and $\Upsilon(x,y)=x^{-\beta}$, yields
\begin{align*}
 \Delta_{\Upsilon\mu}
&= y^2\left(x^{2\beta}
 \partial_{xx}^2 + 2 \rho x^{\beta}  \partial_{xy}^2
+  \partial_{yy}^2\right) - y^2\rho \beta x^{\beta-1} \partial_y\\
&= A - y^2\rho \beta x^{\beta-1} \partial_y
= y^2 \left(\widetilde{A} - \rho \beta x^{\beta-1} \partial_y\right)
=:y^2\widetilde{\Delta_{\Upsilon\mu}}.
\end{align*}
By construction, the operator $\Delta_{x^{-\beta}\mu}$ is symmetric on $\Ll^2(\frac{1}{x^{2\beta} y^2}dx)$, and the time changed operator $\widetilde{\Delta_{x^{-\beta}\mu}}$ is symmetric on 
$\Ll^2\left(\frac{1}{x^{2\beta}}dxdy\right)$.
Note that for $\rho=0$ the weighted Laplace-Beltrami operators coincide with the SABR (and timechanged SABR) generators
$\Delta_{x^{-\beta}\mu}=A$ and $\widetilde{\Delta_{x^{-\beta}\mu}}=\widetilde A$ in \eqref{eq:GeneratorTimeChangedSABR} on their common domain which includes the dense subset $C_0^{\infty}(\Ss)$. The proof of Theorem \ref{Th:SABRDirichletFormTimechanged} follows by the same arguments as presented above.
\end{proof}
\begin{remark}
If $(\mathcal{S},g)$ is a Riemannian manifold with volume element $\mu$ (that is, $\mu=\sqrt{\det g}$) and $\Upsilon:\mathcal{S}\rightarrow [0,\infty)$ is smooth and non-vanishing on $\Ss$, then the triple $(\mathcal{S},g,\Upsilon \mu)$
forms a weighted manifold, cf. \cite[Definition 3.17]{grigoryan09AnalysisManifolds}. In this case, the operator \eqref{eq:LaplBeltramiWeightedManifold} is the Dirichlet Laplace-Beltrami operator of the weighted manifold $(\mathcal{S},g,\Upsilon \mu)$, cf. \cite[Equation (3.45) p. 68.]{grigoryan09AnalysisManifolds}, and in particular $\Delta_{\Upsilon \mu}$ is self-adjoint (\cite[Theorem 4.6]{grigoryan09AnalysisManifolds}) on the weighted Sobolev space $\mathcal{D}(\Delta_{\Upsilon \mu})=W_0^2(\Ss,\Upsilon\mu)$, see \cite[page 104]{grigoryan09AnalysisManifolds}.
\end{remark}

\section{Reminder on diffusions, their geometry and time change}
\subsection{The SABR time change as a change of the underlying geometry}\label{sec:TimechangeGeometry}
It was shown in \cite{Norris}, that the small time asymptotic behaviour of a symmetric elliptic diffusion\footnote{See Appendix \ref{Sec:DiffusionasDirichletform} for a reminder of the definition of a diffusion in the more general context of Dirichlet forms.} (with values on a Lipschitz manifold) is determined by the intrinsic metric (and hence geometry). Moreover, also a converse statement is true: in \cite{SturmDiDetMetr} the question whether a diffusion is determined by its intrinsic metric is answered affirmatively (for diffusions with continuous coefficients). 
In the uniformly elliptic case, 
the intrinsic metric $d(x,y)$ is the Riemannian distance resulting from the Riemannian metric associated with the inverse of the matrix of coefficients of the highest order term of the operator:
\begin{equation*}
d(x,y)=\textrm{inf}\left\{\int_0^1
\sqrt{\langle \dot{\gamma_i}(t),\dot{\gamma_j}(t)\rangle}_{g^{i,j}(\gamma(t))}
dt: \ \gamma\in C^1([0,1]),\ \gamma(0)=x, \gamma(1)=y \right\} 
\end{equation*}
see \cite[Section 5.5.4]{SturmGeometricAspect}, where the term in the integral is the length of the gradient of a (minimal) parametrised curve from $x$ to $y$.
\\
The time change changes the geometry underlying the state space $\Ss$ of the model. We exemplify this statement here in two simple cases ($\beta=0$ and $\beta=1$) of parameter configurations.
It is well-known that the instantaneous variance of a diffusion 
determines the geometry of the process, cf. \cite{GatheralWang, GatheralHsuLawrenceWu,  GulisashviliLawrence, Hsu, LabordereBook, Labordere}, see also \cite{AitSahalia}.
If the instantaneous covariance matrix of the diffusion is non-degenerate on $\Ss$,
then its inverse determines the coefficients of a Riemannian metric (a so-called intrinsic metric) on the state space of the process.
For example for the SABR process \eqref{eq:SABRSDE} with $\beta=0$ and $\beta=1$ (the \emph{normal} and \emph{lognormal} SABR models\footnote{We have not posed any boundary conditions at $x=0$ on the process \eqref{eq:SABRSDE} here.}) this manifold and Riemannian metric are
\begin{equation}\label{eq:ManifoldMetricBeta10}\begin{array}{lllll}
\mathcal{S}_0:=\R\times(0,\infty) \quad&\textrm{and}&\quad g_0(x,y) := \frac{1}{1-\rho^2}\left(\frac{d x^2}{y^2 } - \frac{2 \rho \ d xd y}{y^2}
 + \frac{d y^2}{y^2}\right),
 &\quad (x,y) \in \mathcal{S}_0\\
\mathcal{S}_1:=(0,\infty)^2 \quad&\textrm{and}&\quad g_1(x,y) := \frac{1}{1-\rho^2}\left(\frac{d x^2}{y^2x^2 } - \frac{2 \rho \ d xd y}{y^2 x}
 + \frac{d y^2}{y^2}\right),
 &\quad (x,y) \in \mathcal{S}_1,
\end{array} 
\end{equation}
respectively\footnote{Note that for $\rho=0$, the Riemannian manifold $(S_0,g_0)$ is hyperbolic space $\HH^2$, see \cite[Section 3.1]{HLW}, and the case $\beta=1$ is related to $\HH^3$, cf. \cite[pages 176-178]{LabordereBook}. }. 
Recall that in the non-degenerate (uniformly elliptic) case a unique geometry and intrinsic metric of the diffusion is determined via Varadhan's formula \eqref{eq:VaradhanFormula}, characterising the diffusion (its transition density) at leading order. 
Moreover, the infinitesimal generator of a diffusion of the above type coincides in leading order with the Laplace operator\footnote{More commonly: Laplace-Beltrami operator, see \cite[Equation (3.45) p. 68.]{grigoryan09AnalysisManifolds} for reference.} of the respective manifold. 
The heat equation of the manifold induced by this Laplace operator is solved by the corresponding \emph{heat kernel} 
which determines the law of the \emph{Brownian motion of the manifold}, cf. \cite{Hsu}. The short-time asymptotics of this heat kernel coincide at leading order with those of the transition density $p$ in \eqref{eq:VaradhanFormula},
The geometry of the time-changed process \eqref{eq:SABRdynamicsTimechanged} for the parameters $\beta=0, \beta=1$ is
determined analogously, the manifold and Riemannian metric are
\begin{equation}\label{eq:ManifoldMetricBeta10}\begin{array}{lllll}
\widetilde{\mathcal{S}}_0:=\R^2 \quad&\textrm{and}&\quad \tilde{g}_0(x,y) := \frac{1}{1-\rho^2}\left(d x^2 - 2 \rho \ d xd y
 + d y^2\right),
 &\quad (x,y) \in \widetilde{\mathcal{S}}_1\\
\widetilde{\mathcal{S}}_1:=(0,\infty)\times \R \quad&\textrm{and}&\quad \tilde{g}_1(x,y) := \frac{1}{1-\rho^2}\left(\frac{1}{x^2 }d x^2 - \frac{2 \rho }{x}\ d xd y
 + d y^2\right),
 &\quad (x,y) \in \widetilde{\mathcal{S}}_0,
\end{array} 
\end{equation}
respectively\footnote{The geometry of $(\widetilde{\mathcal{S}}_0,\tilde{g}_0)$ is flat, that is Euclidean.}.
In this simple case it is easy to see that the transition from one geometry to the other is induced---in accord with Theorem \ref{Th:TimechangedSABRVolkonskii}---by the time-change. Later on, we show that this holds in more generality.
In fact, the multiplicative perturbation $y^2$ in \eqref{eq:GeneratorTimeChangedSABR} is the inverse of the density of the hyperbolic volume element ($1/y^2dxdy$), and it is easy to see that in the uncorrelated case we indeed pass by the time change from a Euclidean geometry $(\widetilde{\mathcal{S}}_0,\tilde{g}_0)$ to the geometry of the hyperbolic plane $(\Ss_0,g_0)$.
\subsubsection{Random time change via local times}
\label{Sec:TimeChangeLocalTime}
We refer to \cite[Chapter II]{Borodin} and \cite[Chapter 23]{kallenberg2002foundations} to recall some definitions and notation on scalar diffusions (such as the speed measure $m$, the scale function $s$ and the killing measure $k$), which will be used below to motivate corresponding concepts for Dirichlet forms. With this we aim to give a justification of our statements about the role of local time for the time change \eqref{eq:Thetimechange} in an simple setting.

Recall from \cite[II.13]{Borodin} the local time of a diffusion. 
\begin{definition}[Local time]
For a regular diffusion $X$ the a family of random variables
\begin{equation}
\{L(t,x):x\in I, t\geq 0\}, 
\end{equation}
is called \emph{local time} of $X$, if
\begin{equation}\tag{i}
 \int_0^t\mathbb{I}_A(X_s)ds=\int_A\mathbb{I}_A L(t,x) m(dx) , \quad a.s., A\in \mathbb{B}(I)
\end{equation}
\begin{equation}\tag{ii}
L(t,x)=\lim_{\epsilon \rightarrow 0} \frac{\textrm{Leb}(\{s<t:x-\epsilon<X_s<x+\epsilon\})}{m((x-\epsilon,x+\epsilon))} , \quad a.s.
\end{equation}
\begin{equation}\tag{iii}
 L(t,x,\omega)=L(s,x,\omega)+L(t-s,x,\theta_s(\omega)), \quad a.s.,
\end{equation}
where in $(iii)$, $s<t$ and $\theta$ denotes the shift operator.
\end{definition}
The following connection between local time and the transition density $p$ of $X$ 
holds true:
\begin{equation*}
E^x[L(t,y)]=\int_0^tp(s,x,y)ds. 
\end{equation*}
Now consider a diffusion $X$ on a natural scale (recall: $s(x)=x$ for the scale function) and let $m$ denote the speed measure of $X$. Furthermore, let $W$ be a Brownian motion on $\R$ with $X_0=W_0$. The construction of a random time change via local times presented in \cite[II.16]{Borodin} enables us to obtain a version of $X$ from $W$, which we recall here. See \cite[II.16  and II.21.]{Borodin} for full details. Let now $\{L(t,x), x \in \R, t\geq 0 \}$ denote the local time of $W$. We set 
\begin{equation}
L^{(m)}(t):=\int_{\R}L(t,y)m(dy). 
\end{equation}
Furthermore, let us assume $m(dx)=2m(x)dx$. Then
\begin{equation}
L^{(m)}(t)=\int_0^{\infty}m(W_u)du. 
\end{equation}
Setting
\begin{equation}\label{eq:TimechangeLocaltime}
\tau^{(m)}(t)=\inf\{u:L^{(m)}(u)>t\}. 
\end{equation}
If the boundaries $l,r$ are absorbing then the random time change of the Brownian motion $W$ based on $L^{(m)}$ coincides in law with $X$:
\begin{equation}\label{eq:BrownianLocalTimeChange}
W^{(m)}:=\{W_{\tau^{(m)}(t)}:t>0\} \sim X.
\end{equation}
Furthermore, also the local times of $W^{(m)}$ and $X$ with respect to the speed measure $m$ coincide in law:
\begin{equation}
\{L_X(t,x):t\geq 0,x\in I\} \sim \{ L(\tau^{(m)}(t),x),t\geq0, x\in I\}.
\end{equation}
where $\{L_X(t,x)\}$ denotes the local time of $X$ with respect to the speed measure $m$.
\begin{remark}[Motivation for "speed measure", Borodin-Salminen]\label{rem:GeomPerspLocalTimeAppendix}
The definition of local time $L^{(m)}$ with respect to the speed measure displays that if a Brownian particle is moving
at time $t$ on the region where speed the measure $m$ takes large values, then $L^{(m)}$ is increasing rapidly. 
By \eqref{eq:TimechangeLocaltime} the time change behaves as $\tau^{(m)}(t + h) \simeq \tau^{(m)}(t)$ even for large $h > 0$, and therefore the increment $W_{\tau^{(m)}(t+h)}-W_{\tau^{(m)}(t)}$ is ``small" and hence the increment
$X_{t + h} - X_{t}$ as well. 
Due to this property $m$ has been given the name speed-measure. 
See \cite[II.16]{Borodin} for full details.
See also \cite[Theorem 23.9]{kallenberg2002foundations} and  Remark \ref{rem:GeomPerspLocalTime} on this matter. 
\end{remark}
\subsection{Symmetric Dirichlet forms: closability and the energy measure}
Among strong Markov processes on a domain $U\subset \R^d$, $d\geq1$ with continuous sample paths there is a special class (Hunt processes, see \cite[Appendix A.2]{FukushimaOshimaTakeda} for a precise definition) for which there is a well known correspondence with symmetric Dirichlet forms on Hilbert spaces $\Ll^2(U,m)$, where $m$ is a positive Radon measure with $\supp(m)=\R^d$, the \emph{speed measure}. 
\subsubsection{Beurling-Deny-LeJan Formulae}\label{Sec:BeurlingDenyLeJan}
In case the jump measure is vanishing, every symmetric Markovian form $\cE$ on $\Ll^2(U,m)$ with $\DcE=C^{\infty}_0(U)$ which is closable in  $\Ll^2(U,m)$, can be expressed uniquely (cf. \cite[Theorem 3.2.3]{FukushimaOshimaTakeda} and \cite[Example 1.2.1]{FukushimaOshimaTakeda}) as
\begin{equation}\label{eq:DirichletformBeurlingDeny}
\cE(f_1,f_2)=\sum_{i,j=1}^d\int_{U}\partial_if_1\partial_jf_2 \ \mu_{i,j}(dx)+\int_{U}f_1f_2 \ k(dx), 
\end{equation}
where $(\mu_{i,j})_{i,j}$ is a non-negative definite matrix of Radon measures on $U$, i.e. for any $\xi \in \R^d$ and any compact set $K\subset U$
\begin{equation*}
\sum_{i,j=1}^d\xi_i\xi_j\mu_{ij}(K)\geq 0 \quad \textrm{and}\quad \mu_{i,j}(K)= \mu_{j,i}(K) \quad \textrm{for all } \ 1\leq i,j\leq d,
\end{equation*}
and $k$ is a positive Radon measure, the \emph{killing measure}.
The crucial property for the reverse conclusion---i.e. whether it is possible to construct an associated Hunt process to \eqref{eq:DirichletformBeurlingDeny}---is the closability of the form, see Appendix \ref{Sec:Closability}.

\subsubsection{Symmetric Dirichlet forms and Closability}\label{Sec:Closability}
Closability in the univariate case is completely solved (cf. \cite[Section 1]{RoecknerWielens}), see Theorem \ref{Th:HamzaCondition} for closability conditions for $k=0$, see also Remark \ref{Rem:EngelbertSchmidt} for the corresponding conditions for scalar diffusions.
For time-changes of Dirichlet forms by additive functionals see \cite[Chapter 5 and Section 6.2]{FukushimaOshimaTakeda} and for their closability under a change of speed measure see \cite[Section 5]{RoecknerWielens}.
\begin{definition}[Closability]\label{DefiClosability}
A form $(\cE, \DcE)$ is called \emph{closable} if it has a closed extension. This is equivalent to the condition
\begin{align}\label{EqDefiClosability}
\begin{split}
&u_n \in \DcE, \ \cE(u_n-u_m,u_n-u_m)\longrightarrow_{n,m\rightarrow\infty}0 \\
&(u_n,u_n) \longrightarrow_{n\rightarrow\infty}0 \quad \Rightarrow \quad \cE(u_n,u_n) \longrightarrow_{n\rightarrow\infty}0.
\end{split}
\end{align}
\end{definition}

\subsubsection{Univariate case}
In general, not all non-negative definite symmetric bilinear forms on $C^{\infty}_0(\R)$ are closed or even closable.
In the one-dimensional case \cite[Chap.3 \S 3.1 ($3^\circ$), p. 105]{FukushimaOshimaTakeda}
and \cite[Section 2. p.405.]{albeverio-rockner} give a precise condition for closability of a form 
\begin{equation}\label{eq:UnivariateDirichletForm}
\begin{array}{ll}
 \cE(u,v)&=\int_{-\infty}^{\infty}u'(x)v'(x)\nu(dx)\\
\DcE&=C_0^{\infty}(\R),
\end{array}
\end{equation}
on a Hilbert space $L^2(\R,m)$ for a positive Radon measure $m$ in terms of the singular set of $\nu$.
\begin{theorem}[Hamza condition for closability]\label{Th:HamzaCondition}
The form \eqref{eq:UnivariateDirichletForm} is closable in the Hilbert space $L^2(\R,m)$ for a positive Radon measure $m$ if and only if the following conditions are satisfied:
\begin{itemize}
 \item[(i)] $\nu$ is absolutely continuous (i.e. $\nu(dx)=\rho(x)dx$)
 \item[(ii)] The density function $\rho$ a vanishes a.e. on its singular set.
\end{itemize}
\end{theorem}
\begin{proof}
See For the one-dimensional case \cite[Chap.3 \S 3.1 ($3^\circ$), p. 105]{FukushimaOshimaTakeda} 
and \cite[Section 2. p.405.]{albeverio-rockner}.
\end{proof}
\begin{definition}[Regular and singular sets, univariate case]\label{Regularset1dim}
Given a Borel-measurable function
$\rho:\R \rightarrow \R^+$, 
$x \in \R$ is a 
\emph{regular point} of $\rho$, if there exists an $\epsilon >0$ such that
\begin{align}\label{eq:DefRegularSet}
\int_{x-\epsilon}^{x+\epsilon}\tfrac{1}{ \left|\rho(x)\right|} dx< \infty, 
\end{align} 
and \emph{regular set} of $\rho$ is defined as the (open) set of regular points in $\R$,
denoted by
$R(\rho )$. 
The complement is called the \emph{singular set} and it is denoted
by $S(\rho):=\R \setminus R(\rho )$.
\end{definition}
\begin{remark}\label{Rem:EngelbertSchmidt}
Note that if  [(i)] is fulfilled and the weight $\rho:\R \rightarrow \R_{\geq 0}$ in \eqref{eq:UnivariateDirichletForm} is Borel-measurable and such that 
$\rho=0$ ds-a.e. on the singular set $S(\rho)$ and $\rho>0$ ds-a.e. on the regular set $R(\rho)$.
Then $L^2(R(\rho), \rho \ ds) \subset L^1_{loc}(R(\rho),ds)$ continuously.
In particular, this is the case if $\rho$ is a power-type weight (or a more generally a Muckenhaupt weight).
Note also, that the \emph{Hamza condition} coincides with the \emph{Engelbert and Schmidt condition} \cite[Theorem 23.1]{kallenberg2002foundations} for scalar diffusions. 
There, the corresponding conclusion if the condition is satisfied, is that weak existence for the SDE holds if [(ii)] if \ref{Th:HamzaCondition} holds and uniqueness in law holds for every initial distribution if and only if the density function vanishes nowhere else but on the singular set $S(\rho)$.
\end{remark}
\begin{lemma}[Domain of closedness]\label{Th:OneDimClosability}
Let $\rho$ be as in Theorem \ref{Th:HamzaCondition}. Consider the set 
\begin{equation}\label{eq:DomainofClosedness1dim}
\begin{array}{ll}
D(\cE)=\{&u \in L^2(\R, \rho ds): \ \exists\textrm{ an abs. cont. version }  \\
& \textrm{of} \ u \ \textrm{on} \ R(\rho), \ \textrm{s.th.} \ \frac{du(s)}{ds} \in L^2(\R, \rho ds) \quad \} 
\end{array}\end{equation}
together with the bilinear form
\begin{align*}
 \cE(u,v)=\int_{\R}\frac{u(s)}{ds}\frac{v(s)}{ds}\rho(s)ds.
\end{align*}
Then $(\cE,D(\cE))$ is a closed form on the Hilbert space $L^2(\R, \rho ds)$.
\end{lemma}
\begin{proof}
See \cite[Section 2]{albeverio-rockner}, in particular Condition $(H)$ and Theorem 2.2 $(i)$. 
\end{proof}
\subsubsection{Multivariate case}
The following theorem gives conditions for closability of a non-negative definite symmetric bilinear form $\cE$ in the multivariate case and also specifies the domain $\DcE$
where it is a Dirichlet form (i.e. closed). The theorem can be found in \cite[Proposition 1]{BouleauDenis} or \cite[Section 4]{RoecknerWielens}.
For any $d \in \mathbb{N}$, we denote by $\cB (\R^d )$ the
Borel sigma-field on $\R^d$ and by $dx$ the Lebesgue measure
on $(\R^d,\cB (\R^d))$. 
\begin{definition}[Regular and singular sets, multivariate case]\label{Regularset}
For any $\cB (\R^d )$,  $d\in \mathbb{N}$ measurable function $\rho:\R^d \rightarrow \R^+$,
let $R(\rho )$ denote the \emph{regular set} of $\rho$, i.e. the largest open set in $\R^d$ on which $\rho^{-1}$ is
locally integrable:
\begin{equation*}
 \int_K \tfrac{1}{\left|\rho(x)\right|} dx< \infty \quad \forall K\subset R(\rho),
\end{equation*}
where 
$K\subset R(\rho)$ denotes a compact set in $R(\rho)$. 
Similarly to the univariate case, the complement is called the \emph{singular set} and is denoted
by $S(\rho):=\R^d \setminus R(\rho )$.
\end{definition}
\begin{definition}
Let $u:\R^d\rightarrow \R$ be any $\cB (\R^d )$-measurable function, where $d>0
\in \mathbb{N}$. 
For any $i\in \{ 1,\ldots ,d\}$ with corresponding
$\bar x:=(x_1 ,\ldots , x_{i-1}, x_{i+1},\ldots,x_{d})\in \R^{d-1}$,
consider the function $u_{\bar x }^{(i)}$ defined by 
$u_{\bar x }^{(i)}:\R \rightarrow \R,\ s \mapsto u ((\bar x,s)_i),$
where 
$(\bar x ,s)_i:=(x_1,\ldots x_{i-1},s,x_{i+1},\ldots ,x_{d} ).$
\end{definition}
\begin{definition}[The bilinear form and its domain]\label{Defibilinearform}
Let the positive Borel function 
$m:\R^d\rightarrow \R^+$ denote the speed measure and consier
an $\R^{d\times d}$-valued symmetric Borel function 
\begin{align}\label{MatrixForm}
\begin{split}
\xi : \R^d& \longrightarrow \R^{d\times d}\\
      x:=(x_1,\ldots,x_d)&\longmapsto (\xi_{ij}(x))_{1\leq i,j\leq d}.
\end{split}
\end{align}
\begin{itemize}
 \item 
Consider the following  \emph{symmetric bilinear form} on $\Ll^2 (m dx)$: 
\begin{equation}\label{BilinearformClosability}
\mathcal{E}(u,v)=\frac12 \int_{\R^d}\sum_{i,j}
\xi_{ij}(x) \partial_i u(x)\partial_j v (x) m(x)\, dx. 
\end{equation}

\item
Define the \emph{domain of $\cE$} as the set\\ 
$$D(\cE)=\Big\{u\in\Ll^2(mdx)\quad \cB (\R^d)\textrm{-mb: for all } i\in\{ 1,\ldots ,d\} \textrm{ and } \lambda^{d-1} \textrm{-almost all }\bar x\in\R^{d-1}$$ 
\begin{center}$u^{(i)}_{\bar x}$ has an absolutely continuous\footnote{
See for example:\cite{albeverio-rockner} p.406.}
version $\widetilde{u}^{(i)}_{\bar x}$ on $R(m_{\bar x}^{(i)})$, s. th. \end{center} $$\sum_{i,j}\xi_{ij}\frac{\partial
u}{\partial x_i } \frac{\partial u}{\partial x_j } \in \Ll^1 (mdx), \ \textrm{where }\frac{\partial u}{\partial x_i }:=\frac{d\widetilde{u}^{(i)}_{\bar x}}{ds}. \Big\}$$
\end{itemize}
\end{definition}

\begin{condition}[HG1]\label{HamzaGeneral1}
$m_{\bar x}^{(i)}=0$, $\lambda^1$-a.e. on
$\R\setminus R(m_{\bar x}^{(i)})$, for any $i\in\{1,\cdots ,d\}$ and
$\lambda^{d-1}$-almost all $\bar x\in\{y\in\R^{d-1}:\ \int_{\R}
m^{(i)}_y (s)\, ds >0\}$.
\end{condition}
\begin{condition}[HG2]\label{HamzaGeneral2}
There exists an open set $O\subset \R^d$ such that $\lambda^d (\R^d \setminus
O)=0$ and $\xi$
is locally elliptic on $O$ in the sense that for any compact subset
$K$, in $O$, there exists a positive constant $c_K$ such that
\[\forall x=(x_1 ,\cdots ,x_d )\in K,\ \sum_{i,j=1}^d \xi_{ij} (x)x_i x_j \geq
c_K |x|^2 .\]
\end{condition}

\begin{definition}[Energy measure]\label{Def:CarreDuChamp}
Let $\cE$ be a Dirichlet form on $\Ll^2(X,m)$ with Domain $\DcE$. We say that $\cE$ admits a \emph{carr\'{e} du champ operator}
also \emph{square field operator} or \emph{energy measure} if the following property holds:
\begin{enumerate}
 \item[(R)] There exists a subspace $H$ of $\DcE \cap \Ll^{\infty}(X)$, dense in $\DcE$ such that
\begin{align*}
 \textrm{For all} \ f \in H \quad \exists \ &\tilde{f}\in L_1(X), \ \textrm{such that for all} \ 
h\in \DcE \cap \Ll^{\infty}(X)\\
&2\cE(fh,f)-\cE(h,f^2)=\int_Xh\tilde{f}dm.
\end{align*}
\end{enumerate}
If $\tilde{f},\tilde{g}\in \Ll^1(X)$ as above, we define by polarisation a form $\Gamma:\DcE \times \DcE \rightarrow \Ll^1$
\begin{align}\label{DefCarreduChampPol}
\Gamma(f,g):=\frac{1}{4} ((\tilde{f}+\tilde{g})-(\tilde{f}-\tilde{g})).
\end{align}
We refer to $\Gamma$ as the \emph{carr\'{e} du champ operator} associated with $\cE$.
\end{definition}
\begin{remark}[Characterizing property of the carr\'{e} du champ operator]
If (R) is satisfied, then the form $\Gamma$ defined in \eqref{DefCarreduChampPol} is the unique positive symmetric continuous
bilinear form  $\Gamma:\DcE \times \DcE \rightarrow \Ll^1$ with the characterizing property 
that for all $f,g,h \in \DcE \cap \Ll^{\infty}(X)$
\begin{align}\label{CarreduChampCharProp}
\cE(fh,g)+\cE(gh,f)-\cE(h,fg)=\int h \ \Gamma(f,h)dm.
\end{align}
Note that for $f=g$ we are in the situation $(R)$ hence for the $\tilde{f}$ there, we can use the notation $\tilde{f}=\Gamma(f,f)$.
See: \cite{BouleauHirsch} Proposition 4.1.3.
\end{remark}

\begin{theorem}[Conditions for closability and domain of closedness, R\"{o}ckner-Wielens and Bouleau-Denis]\label{closability}
Let $\cE$ denote the bilinear form with domain $D(\cE)$ in (\ref{BilinearformClosability})
satisfying conditions \ref{HamzaGeneral1} and \ref{HamzaGeneral2}. Then the pair 
$(\mathcal{E}, C_0^{\infty}(\R^d))$ is closable on $\Ll^2 (m dx)$ and
$(\mathcal{E},D(\cE))$ is a Dirichlet form on $\Ll^2 (m dx)$ which admits a
carr\'e du champ operator $\Gamma$ given by  
\begin{equation*}
\Gamma [u,v]=\sum_{i,j} \xi_{ij} \partial_i u\partial_j v, \quad \textrm{for }\ u,v\in D(\cE).
\end{equation*}
\end{theorem}
\begin{proof}
The proofs can be found in \cite[Section 4]{RoecknerWielens} for the former statement and \cite[Proposition 1]{BouleauDenis} for the latter. We briefly elaborate on the energy measure:
As remarked in \eqref{CarreduChampCharProp}, the carr\'{e} du champ operator $\Gamma[u,u]$
is characterized by the identity
$2\cE(uv, u)-\cE(u^2,v)=\int_{\R^d}\Gamma[u,u](x) \ v(x) m(x) dx$ for any $u, v \in \DcE$.
\begin{equation*}\begin{array}{ll}
&2\cE(uv, u)-\cE(u^2,v)\\
 &= \int_{\R^d}\sum_{i,j}\xi_{ij}(x) \partial_i( u v)\partial_j (u) m(x)\, dx -
\frac12 \int_{\R^d}\sum_{i,j}\xi_{ij}(x) \partial_i( u^2)\partial_j (v) m(x)\, dx\\
&= \int_{\R^d}\sum_{i,j}\xi_{ij}(x) \partial_iu(x)\partial_j u(x) \  v(x) m(x)\, dx,
\end{array}\end{equation*}
where the last step follows by symmetry of $\xi_{i,j}$.
\end{proof}

\subsection{Diffusions as symmetric Dirichlet forms and the intrinsic metric}\label{Sec:DiffusionasDirichletform}

Let $X$ be a connected locally compact separable metric space with a positive Radon measure $\mu$ with 
$\textrm{supp}(\mu)=X$. Consider the Hilbert space $L^2(X,\mu)$ of real-valued functions, which are
square-integrable with respect to $\mu$. 
Furthermore, let $\cE$ be a Dirichlet form $ L^2(X,\mu)$, i.e. a closed and Markovian, non-negative 
definite symmetric bilinear form $\cE$ on a dense subspace $D(\cE)\subset L^2(X,\mu)$.

\begin{definition}[Regular Dirichlet form]\label{RegularDirichletForm}
A Dirichlet Form $\cE$ in the Hilbert space $\mathcal{H}$ is called \emph{regular} if there is a subset of $D(\cE) \cap C_c(X)$, 
which is a \emph{core} of $\cE$, i.e., which is dense in 
$D(\cE)$ with respect to the natural norm $\varphi\mapsto(\cE(\varphi)+||\varphi||_{\mathcal{H}})^{\frac{1}{2}}$, and which is
dense in $C_0(X)$ with respect to the supremum norm $L^{\infty}$.
\end{definition}
\begin{definition}[Strongly local quadratic form]\label{Stronglylocal}
Let $\cE$ be any positive quadratic form on a Hilbert space $\mathcal{H}$, we call $\cE$ \emph{strongly local} if 
$\cE(\psi,\varphi)=0$ for all $\varphi, \psi \in D(\cE)$ with $\textrm{supp}(\psi)$ and 
$\textrm{supp}(\varphi)$ compact with $\psi$ constant on a neighbourhood of $\textrm{supp}(\varphi)$.
\end{definition}

\begin{definition}[Diffusion]\label{Diffusion}
We call a strongly local regular Dirichlet form a \emph{diffusion}. 
\end{definition}

\begin{definition}[Normal contraction]\label{NormalContraction}
A map $F:\R^n \rightarrow \R$ is called a \emph{contraction} on $\R^n$ 
\begin{align}\label{ConditionContraction}
 |F(x)-F(y)|\leq \sum_{i=1}^{n}|x_i-y_i|,
\end{align}
for all $x,y \in \R^n$. $F$ is called a \emph{normal contraction} (denoted by $F \in \mathcal{T}_0^n$), if $F(0)=0$.
\end{definition}
\begin{lemma}
Let $\cE$ be a Dirichlet form on $L^2(X,\mu)$ with Domain $\DcE$. Let $ \varphi \in \DcE \cap L^{\infty}(X)$ be a non-negative function,
$\psi_1 \in \DcE \cap L^{\infty}(X)$ and $\psi_2=F\circ \psi_1$ for a normal contraction\footnote{See Definition \ref{NormalContraction} above cf. 
\cite[Def.2.3.2.]{BouleauHirsch}. } 
$F \in \mathcal{T}^1_0$. 
\begin{align*}
0 \leq \cE(\psi_2 \varphi, \psi_2)-\tfrac{1}{2}\cE(\psi_2^2,\varphi) 
\leq \cE(\psi_1 \varphi, \psi_1)-\tfrac{1}{2}\cE(\psi_1^2,\varphi) 
\leq ||\varphi||_{L^{\infty}} \cE(\psi_1)
\end{align*}
\end{lemma}
\begin{proof}
The statement follows from Propositions 2.3.3 and 4.1.1. of \cite{BouleauHirsch}.
\end{proof}
In particular, $\cE(\psi_1 h, \psi_1)-\tfrac{1}{2}\cE(\psi_1^2,h)<\infty$. This ensures that the following map is well defined:
\begin{definition}
Let $\cE$ be a Dirichlet form on $L^2(X,\mu)$. Define for any $\psi \in D(\cE) \cap L^{\infty}(X)$, a map 
\begin{align}\label{DefCarreDuChampHalf}
\begin{split}
\mathcal{I}_{\psi}^{(\mathcal{E})}: \DcE \cap  L^{\infty}(X)& \longrightarrow \R \\
                                        \varphi &\longmapsto \cE(\psi \varphi, \psi)-\tfrac{1}{2}\cE(\psi^2,\varphi).
\end{split}
\end{align}
\end{definition}

\begin{definition}[$\DcEl$]\label{DlocE}
Let $\cE$ be a diffusion. Define $D_{loc}(\cE)$ as the vector space of equivalence classes of measurable functions
$\psi:X \longmapsto \mathbb{C}$, such that for every compact subset $K\subset X$ there exists a $\hat{\psi} \in \DcE$
with $\psi|_K=\hat{\psi}|_K$. 
\end{definition}
Now we define the extension of $\I$ from $\DcE \cap L^{\infty}(X)$ to $\DcEl \cap L^{\infty}(X)$ for diffusions.
\begin{definition}[$\widehat{\I}$]\label{DefCarreDuChampHalfHat}
Let $\cE$ be a diffusion. Let us write $\I_{\hat{\psi}}(\varphi)=\I^{(\cE)}_{\hat{\psi}}(\varphi)$ for shorter notation
when $\cE$ is fixed. Let $L^{\infty}_c(X)$ denote the bounded functions on $X$ with compact support.
Then for any $\psi \in \DcEl \cap L^{\infty}(X)$ define the map
\begin{align*}
\widehat{\I}_{\psi} : \DcE \cap L^{\infty}_c(X) & \longrightarrow \R\\
                                    \varphi & \longmapsto \I_{\hat{\psi}}(\varphi).
\end{align*}
\end{definition}
\begin{definition}
For any $\psi \in \DcEl \cap L^{\infty}(X)$ we define
\begin{align}
|||\widehat{\I}_{\psi}|||:=\sup\{|\widehat{\I}_{\psi}(\varphi)|: \varphi \in \DcE \cap L^{\infty}_c(X), ||\varphi||_{L^1(X)}\leq 1 \}
\end{align}
\end{definition}

\begin{definition}[$d_{\psi}(A,B)$ ]\label{SetSetDistanceOnePSi}
Let $\psi \in L^{\infty}(X)$ be an arbitrary bounded function and $A,B \subset X$ measurable sets. We define 
a distance function, taking values in $(-\infty,\infty]$ by
\begin{align}\label{SetSetDistanceOnePSiEq}
\begin{split}
d_{\psi}(A,B):=&\sup\{M \in \R: \psi(x)-\psi(y)\geq M\ \textrm{for}\ \lambda\textrm{-a.e.} x,  \lambda\textrm{-a.e.} y\}\\
              =&\esinf_{x\in A}\psi(x)-\essup_{y\in B}\psi(y),
\end{split}
\end{align}
where $\essup_{y\in B}\psi(y)= \inf\{m\in \R: \lambda(\{y \in B: \psi(y)>m\})=0\}$, and where $\lambda$ denotes the Lebesgue measure.
\end{definition}
With this, one can define the intrinsic metric induced by the Dirichlet form $\dE(A,B)$:
\begin{definition}[\emph{ter Elst et al.}]\label{Def:IntrinsicMetricSetSetDistance}
The set-theoretic distance function for measurable sets $A,B \in X$, 
which appears in the generalized version of Varhadhan's formula is
\begin{align}\label{SetSetDistanceEq}
\dE(A,B)=\sup\{d_{\psi}(A,B): \psi \in \DocE \},  
\end{align}
where the set $\DocE$ is defined as
\begin{align}\label{ExtendedDomain}
D_0(\cE)=\{ \psi \in D(\cE)_{loc} \cap L^{\infty}(X) : |||\hIp||| \leq 1\}. 
\end{align}
\end{definition}


\end{document}